\documentclass[12pt,a4paper]{article}

\usepackage[utf8]{inputenc}
\usepackage{amsthm,amsmath,latexsym,amssymb,amsfonts,amscd,cite}
\usepackage{amssymb} 
\usepackage[all]{xy}
\usepackage{tikz}
\usepackage{array}
\usepackage{hyperref}

\usepackage{xcolor}
\usepackage{tocbasic}
\DeclareTOCStyleEntry[
  beforeskip=.2em plus 1pt,
  pagenumberformat=\textbf
]{tocline}{section}

\usepackage{mathrsfs}
\usepackage{hyperref}

\usepackage{epigraph}

\usepackage{bbm}
\usepackage{bm}
\usepackage[T1]{fontenc}
\usepackage{mathrsfs}

\numberwithin{equation}{section}

\theoremstyle{theorem}

\newtheorem{theorem}{Theorem}[]

\newtheorem{proposition}[]{Proposition}
\newtheorem{definition}[equation]{Definition}

\theoremstyle{remark}
\newtheorem{remark}[equation]{Remark}
\newtheorem{example}[equation]{Example}

\usepackage[all]{xy}
\usepackage{graphicx}

\usepackage{mathtools}

\makeatletter
\DeclareRobustCommand\widecheck[1]{{\mathpalette\@widecheck{#1}}}
\def\@widecheck#1#2{%
    \setbox\z@\hbox{\m@th$#1#2$}%
    \setbox\tw@\hbox{\m@th$#1%
       \widehat{%
          \vrule\@width\z@\@height\ht\z@
          \vrule\@height\z@\@width\wd\z@}$}%
    \dp\tw@-\ht\z@
    \@tempdima\ht\z@ \advance\@tempdima2\ht\tw@ \divide\@tempdima\thr@@
    \setbox\tw@\hbox{%
       \raise\@tempdima\hbox{\scalebox{1}[-1]{\lower\@tempdima\box
\tw@}}}%
    {\ooalign{\box\tw@ \cr \box\z@}}}
\makeatother

\newcommand{\R}{\mathcal{R}}
\newcommand{\hh}{\mathcal{H}} 
\newcommand{\hn}{\mathcal{H}_N} 
\newcommand{\fg}{\mathcal{F}(G)}

\newcommand{\gs}{\mathfrak{g}^\ast} 
\newcommand{\fa}{\mathfrak{A}}

\newcommand{\pol}{\mathrm{Pol}(\mathfrak{g}^\ast )} 
 
\newcommand{\pold}{\mathrm{Pol}(D )} 
\newcommand{\qq}{\mathcal{Q}_{\hbar}} 
\newcommand{\ff}{\mathcal{F}}
\newcommand{\cc}{\mathbb{C}}

\newcommand{\I}{1\!\! 1}

\title{{\bf Quantization of bounded symplectic domains associated with compact Lie groups}}

\date{}
\begin{document}

\maketitle

\begin{center}
\vskip -0.05\textheight
\renewcommand{\thefootnote}{\fnsymbol{footnote}}
  Alexey A. \textsc{Sharapov}\footnote{{\tt sharapov@phys.tsu.ru}} 
\renewcommand{\thefootnote}{\arabic{footnote}}
\setcounter{footnote}{0}

\end{center}
\begin{center}
\emph{Physics Faculty, Tomsk State University,}
\emph{Lenin ave. 36, Tomsk
634050, Russia}
\end{center}

\vspace{5mm}

\begin{abstract}
We present a systematic quantization scheme for bounded symplectic domains of the form 
\(D \times G \subset T^\ast G\), where \(D \subset \mathfrak{g}^\ast\) is a bounded region defined by algebraic inequalities and \(G\) is a compact Lie group with Lie algebra \(\mathfrak{g}\).  
The finiteness of the  symplectic volume implies that quantization yields a finite-dimensional Hilbert space, with observables represented by Hermitian matrices, for which we provide an explicit realization.  
Boundary effects necessitate modifications of the standard von Neumann and Dirac conditions, which usually underlie the correspondence principle. Physically, the compact group $G$ plays the role of momentum space, while $\mathfrak{g}^\ast$
corresponds to the (noncommutative) position space of a particle.
The assumption of compact momentum space has profound physical consequences, including the supertunneling phenomenon and the emergence of a maximal fermion density.
\end{abstract}

\section{Introduction}

Quantization of symplectic manifolds is a recurring theme in mathematical physics. Such manifolds naturally serve as phase spaces for classical mechanical systems, and their quantization provides a bridge from classical to quantum mechanics. The most familiar examples arise when the phase space is the cotangent bundle $T^\ast M$ of a configuration space $M$, whose topology may already be nontrivial. An even richer class of examples comes from constrained Hamiltonian dynamics, where the reduced (physical) phase spaces often exhibit nontrivial global geometry. The multitude of approaches to the quantization problem -- commonly labeled as {\it geometric}, {\it deformation}, {\it strict}, {\it Berezin--Toeplitz},  etc. -- makes any comprehensive review impossible here; we therefore refer the reader to the standard monographs \cite{woodhouse1992geometric, KM, fedosov1995deformation,landsman2012mathematical,floch2018brief}.

The present work is motivated by a different class of examples, where the nontrivial topology of the phase space originates in momentum rather than   configuration space. This shift of perspective introduces new features in the quantization problem and reveals novel connections between geometry and quantum theory. Historically, the hypothesis of curved momentum space was first put forward by Born \cite{Born,born1949reciprocity} and subsequently explored by many authors (see, e.g., \cite{Golfand:1962kjf,Tamm,Kadyshevsky:1977mu,Amelino-Camelia:2011lvm,KOWALSKI_GLIKMAN_2013,Franchino-Vinas:2023rcc} and references therein). One of the original motivations was the idea that a nontrivial momentum space geometry could act as an effective cut-off, mitigating or even eliminating ultraviolet divergences in quantum field theory. 
For instance, if momentum space is compact, then its finite volume provides a natural ultraviolet cut-off.

It was soon realized, however, that the hypothesis of curved momentum space is closely related to the idea of noncommutative spacetime, where particle positions $x^\mu$, represented by  noncommuting operators, cannot be measured with arbitrary precision. In this picture, spacetime itself acquires a noncommutative structure, replacing sharp localization with a ``smeared'' notion of a point.
The proposal to suppress ultraviolet divergences through spacetime noncommutativity was first formulated by Snyder \cite{Snyder1,Snyder2}. 

While these early ideas laid the groundwork, renewed interest emerged around the turn of the millennium with the work of Seiberg and Witten \cite{Seiberg_1999}, who studied open strings ending on $D$-branes in a constant $B$-field background. They demonstrated that the low-energy effective theory on a single $D$-brane admits a consistent formulation as a noncommutative gauge theory. This insight sparked an extensive body of research on noncommutative field theory (see, e.g., \cite{Douglas:2001ba, Szabo:2001kg, Hersent:2022gry}) and was accompanied by parallel developments in mathematics, including the theory of deformation quantization, quantum groups, and the broader framework of noncommutative geometry.

By analogy with conventional quantum mechanics, one expects that spacetime noncommutativity in the semiclassical limit can be captured by suitable Poisson brackets.
Applying this semiclassical approximation to noncommutative gauge theories naturally leads to what is known as Poisson gauge theory \cite{kupriyanov2021poisson,Kupriyanov_2021,Kupriyanov:2022ohu,Kupriyanov:2023gjj},  bridging classical geometry  and its quantum deformations. 
Once the spacetime manifold $M$ is endowed with a Poisson structure, it is natural to ask about the phase space of a point particle moving in $M$. Following the correspondence principle, this phase space should be a symplectic manifold of dimension $2\dim M$, with Poisson brackets extending those on $M$. 

A careful analysis of the minimal coupling of a charged point particle to gauge fields was presented in Refs.~\cite{kupriyanov2024symplectic,sharapov2024poisson}. This work shows that the natural candidate for the phase space is not an ordinary cotangent bundle but a {\it symplectic groupoid}
$
\mathcal{G} \rightrightarrows M
$
integrating the Poisson manifold $M$. As a concrete example, let $M=\mathbb{R}^{d-1,1}$ be Minkowski space endowed with linear Poisson brackets 
\begin{equation}\label{LPB}
\{x^\mu, x^\nu\} = f^{\mu\nu}{}_\lambda \, x^\lambda\,,
\end{equation}
so that the dual space naturally acquires the structure of a Lie algebra $\mathfrak{g}$ with structure constants $f^{\mu\nu}{}_\lambda$. The corresponding symplectic groupoid $\mathcal{G}$ over $M$ is then given by the cotangent bundle 
$
T^\ast G
$,
where $G$ is a Lie group integrating the Lie algebra $\mathfrak{g}$. In this construction, the Poisson brackets arise from the canonical symplectic structure on $T^\ast G$, thereby providing a {\it symplectic realization} of the linear Poisson structure on $M$. Using left-invariant vector fields, one identifies $T^\ast G \simeq \mathfrak{g}^\ast \times G$, where the first factor $\mathfrak{g}^\ast$ corresponds to the configuration space $M$, while $G$ represents the momentum space. 
This example explicitly illustrates how spacetime noncommutativity induces a nontrivial geometry on momentum space, curving it according to the underlying Poisson structure. If $\mathfrak{g}$ is compact semisimple, the associated Lie group $G$ and thus the momentum space are compact, in contrast with conventional classical mechanics, where the momentum space $T^\ast_p M$ at $p\in M$ is always linear.

Several mechanical systems on symplectic groupoids $T^\ast G \rightrightarrows \mathfrak{g}^\ast$ have been analyzed in recent years \cite{FSK,MFK,kupriyanov2024classical, basilio2025charged}. These studies indicate that the compactness of momentum space can lead to notable physical consequences. For instance, the velocity of a free (nonrelativistic) particle is bounded from above, provided that the kinetic energy is a smooth function of the momenta. Furthermore, when a particle is subjected to a central potential, it cannot fall to the center, regardless of the strength of the attractive force. Conversely, for any repulsive central force, there exists a finite domain of bounded motion. These results challenge conventional intuition based on flat momentum spaces. 

Taking noncommutative spacetime as a working hypothesis, one is naturally led to consider the quantization of the classical phase space given by the symplectic groupoid $\mathcal{G} \rightrightarrows M$.
In this work, we focus on a particular instance of this general problem, where the noncommutative configuration space $M$ is modeled by a bounded domain $
D \subset \mathfrak{g}^\ast
$ 
in the dual of a compact Lie algebra, endowed with the linear Poisson bracket~(\ref{LPB}). Equivalently, we consider a nonrelativistic particle confined to a finite region of the noncommutative space $\mathfrak{g}^\ast$. The corresponding momentum space is the compact Lie group $G$, and the resulting phase space is diffeomorphic to the direct product 
$
D \times G
$,
carrying  the symplectic structure induced by the embedding
$
D \times G \hookrightarrow  T^\ast G
$. 
Since this phase space has finite symplectic volume, any consistent quantization must yield a finite-dimensional Hilbert space, where observables are represented by finite Hermitian matrices. This leads to discrete spectra and bounded physical quantities, in sharp contrast with those of standard quantum mechanics.

In the following sections, we develop the necessary mathematical framework for this quantization and discuss
some of its elementary physical consequences.

\section{Quantization on a rectangle and a bounded cylinder}\label{RecDom}
To explain the basic idea, let us start with a simple but instructive example of a rectangular domain on the phase plane with canonical coordinates $x$ and $p$. The domain $\R\subset{\mathbb{R}^2}$ is defined by the inequalities:
\begin{equation}
    0< x \leq L\,,\qquad 0\leq p\leq M\,.
\end{equation}
We interpret $x$ as a position coordinate and $p$ as the conjugate momentum. Thus, both the position and momentum spaces are bounded, and the canonical volume (area) of the domain is  $LM$. 
Classical observables are identified with smooth complex-valued functions on $\R$. These form a Poisson algebra with respect to the canonical Poisson brackets $~{\{p,x\}=1}$. The Poisson algebra $C^\infty(\R,\cc)$  contains a dense subalgebra $\fa$, multiplicatively generated by the functions $x$ and $z=e^{\frac{2\pi i p}{M}}$. As a linear space, $\fa$ is spanned by the monomials
\begin{equation}\label{xz}
    x^n z^m\,,\qquad n\in \mathbb{N}_0\,,\quad m\in \mathbb{Z}\,. 
\end{equation}
By the Stone--Weierstrass theorem, every continuous (and hence smooth) function on $\R$ can be uniformly approximated by finite linear combinations of monomials (\ref{xz}). 

Upon quantization, the state space of the system, denoted $\mathcal{H}_N$, must be finite-dimensional. Indeed, the Bohr--Sommerfeld 
quantization rule implies that in the semiclassical limit, 
\begin{equation}\label{BSR}
    N:=\dim \hn \sim \left[ \frac{LM}{2\pi\hbar}\right]\,, \qquad LM\gg \hbar\,. 
\end{equation}
Here, the square brackets stand for the floor function. 
For our purposes, it is convenient to work in the momentum representation and realize $\hn$ as a subspace of the Hilbert space $\hh$ of complex-valued functions $\psi(p)$. The Hermitian inner product is defined in the standard way:
\begin{equation}
    (\psi,\psi')=\frac1M\int_0^Mdp\bar \psi(p)\psi'(p)\qquad \forall \psi,\psi'\in \hh\,.
\end{equation}
Then the functions $\psi_n=z^n$, $n\in \mathbb{Z}$, form an orthonormal basis for $\hh$:
\begin{equation}\label{base}
    (\psi_n,\psi_m)=\delta_{nm}\,.
\end{equation}
By analogy with canonical quantization, we set
\begin{equation}
    \hat{x}=-i\hbar\frac{d}{d p}\,.
\end{equation}
Integration by parts shows that $\hat{x}$ defines a self-adjoint operator provided that the wave functions satisfy the periodic boundary condition 
\begin{equation}\label{PBC}
    \psi(0)=\psi(M)\,.
\end{equation}
However, this is not the only condition we need to impose on admissible states. Indeed, applying the operator $\hat x$ to the basis functions $\psi_n$, we find 
\begin{equation}
    \hat{x}\psi_n=x_n\psi_n\,,\qquad x_n=\frac{2\pi n\hbar}{M}\,,
    \end{equation}
meaning that $\psi_n$ is an eigenstate of the position operator with the definite coordinate value $x_n$. Since the motion of the particle is confined to the interval $(0, L]$, we arrive at  the inequalities
\begin{equation}
    1\leq n\leq N\,, 
\end{equation}
where $N$ is given by Eq.~(\ref{BSR}), together with
\begin{equation}
    \qquad\frac{LM}{2\pi (N+1)}\leq \hbar \leq \frac{LM}{2\pi N}\,.\end{equation}
Hence, only  the eigenstates $\psi_1,\psi_2,\ldots, \psi_{N}$ are admissible. By the superposition principle, these span a state space $\hn$ of dimension $N$, in agreement with Eq.~(\ref{BSR}). 

Instead of quantizing momentum coordinate $p$, we may consider the quantization problem for Laurent polynomials $f(z)$, which are just specific functions of $p$.  If the state space were the whole $\hh$, we could apply the recipe of canonical quantization and define the action of $\hat f$ on $\psi(p)$ as multiplication by $f(z)$. This definition, however, would generally yield a function $f(z)\psi(p)\notin \hn$. Therefore, we introduce the Hermitian projector $P_N: \hh\rightarrow \hn$ and define the operator $\hat f$ by the rule:
\begin{equation}\label{Pz}
    \hat{f}\psi=P_Nf(z)\psi\,.
\end{equation}
It is easy to see that 
     $\hat f^\dagger= \hat{\bar f}$ and the operator $\hat f$ is Hermitian for real-valued functions $f(z)$. It is convenient to introduce the indicator function 
     \begin{equation}
         \chi_N(n)= \left\{\begin{array}{cc}
              1\,, & 1\leq n\leq N ; \\
               0\,,& \mbox{otherwise} \,.
          \end{array}\right.
      \end{equation}     
Then
\begin{equation}
    \widehat{z^m}\psi_n=\chi_N(n+m)\psi_{n+m}\,.
\end{equation}
Because of the projector in the definition (\ref{Pz}), 
\begin{equation}
    \widehat{z^n}\widehat{z^m}\neq \widehat{z^{n+m}}\neq \widehat{z^m}\widehat{z^n}
\end{equation}
in general. However, these equalities hold whenever $nm\geq 0$. In particular, 
$\widehat{z^n}=\hat{z}{}^n $.
We can now extend this quantization rule to general monomials in the $x$'s and $z$'s as
     \begin{equation}\label{quant}
         \widehat{x^n z^m}=\frac12\big(\hat{x}^n\widehat{z^m}+\widehat{z^m}\hat{x}^n\big)
         \end{equation}
         and then to arbitrary elements of $\fa$ by linearity. It follows from the definition that
         \begin{equation}
             [\hat x,\hat z]=i\hbar \widehat{\{x, z\}}\,,
         \end{equation}
         which is a manifestation of the {\it correspondence principle} in Dirac's sense. However, $~{[\hat{z},\hat{\bar z}]\neq 0 }$, in contrast to the Poisson brackets $\{z,\bar z \}=0$. 
         Formula (\ref{quant}) defines a map 
\begin{equation}\label{QM}
\qq^\R: \fa\rightarrow \mathscr{L}(\hn)
\end{equation}from the Poisson algebra of classical observables to the algebra of operators on a finite-dimensional Hilbert space $\hh_N$. We will refer to $\qq^\R$ as a {\it quantization  map}. Since the algebra $\mathscr{L}(\hn)$ is finite-dimensional, the quantization map is far from being injective and much information about $\fa$ is lost on the operator side. This information can hopefully be restored in the classical limit as $\hbar\to 0$ and $N \to \infty$. 
         To trace the classical limit carefully, we make $\fa$ and $\mathscr{L}(\hn)$ into normed spaces. In both cases, the norm is introduced through a positive-definite inner product. For classical observables, we can use the canonical measure on the rectangle $\R$ to define the inner product
         \begin{equation}
             (a,b)_\R=\frac{1}{LM}\int_\R\overline{ a(x,p)} b(x,p)dxdp\,,\qquad \forall a,b\in \fa\,,
         \end{equation}
         and for the operator algebra, we take the {\it normalized} Frobenius inner product 
         \begin{equation}\label{FN}
             (A,B)_N=\frac1{N}\mathrm{Tr}(A^\dagger B)\,,\qquad \forall A,B\in \mathscr{L}(\hn)\,.
         \end{equation}
The corresponding norms are defined as
         \begin{equation}
             \|a\|_\R^2=(a,a)_\R\,,\qquad \|A\|^2_N=({A},{A})_N\,.
         \end{equation}
They satisfy all standard properties, including positive-definiteness\footnote{It should be stressed that the normalized Frobenius norm is {\it not} submultiplicative; rather, it satisfies the inequality $\|A B\|_N\leq \sqrt{N}\|A\|_N \|B\|_N$. Consequently, $\mathscr{L}(\hn)$ is not a Banach  algebra unless $N=1$. Nonetheless, this norm proves to be extremely useful for  studying the asymptotic behaviour of large matrices \cite{vonNeumann1942}, \cite{ CIT-006 }.\label{f1} }. The normalization is chosen such that $\|1\|_\R=1$ and $\|\I\|_N=1$.  Using an orthonormal basis $\psi_1,\ldots,\psi_N\in \hh_N$, one can also express the normalized Frobenius norm in terms of the usual operator norm:
\begin{equation}
    \|A\|^2_N=\frac1N\sum_{k=1}^N\|A\psi_k\|^2\,.
\end{equation}

\begin{proposition}\label{P1} 
The quantization map (\ref{QM}) satisfies the following properties:
\begin{itemize}
\item[i)] Normalization:
$$
\hat 1=\I\qquad (\mbox{\it identity operator on } \hn)\,,
$$
\item[ii)] Reality:
$$
\widehat{\bar a}=\hat{a}^\dagger\,,
$$
    \item[iii)] {Separability:} $$ \lim_{N\rightarrow \infty}\|\hat a\|_N=\|a\|_\R\,,$$ 
\item [iv)]{ von Neumann's condition:} 
$$
\lim_{N\to\infty}\|\hat a\hat b-\widehat{ab}\|_N=0\,,
$$
\item[v)] Dirac's condition: for any pair $a,b\in \fa$, there exists a subspace $\hh^{a,b}\subset \hh_N$ such that
$$
\lim_{N\rightarrow \infty}\frac{\dim\hh^{a,b}}{N}=1
$$
and
$$
\lim_{N\to\infty}\big\|\big((i\hbar)^{-1}[\hat a, \hat b]-\widehat{\{a, b\}}\big)\psi\big\|=0\qquad \forall \psi\in \hh^{a,b}\,.
$$
\end{itemize}
\end{proposition} 
Properties ($i$) and ($ii$) follow immediately from the definition of the quantization map.
The proof of ($iii$-$v$), which involves more technical arguments, is deferred to Appendix A. Several remarks are now in order. 
\begin{remark}
    Properties ($i$-$iv$) resemble the defining conditions of {\it strict quantization} \cite{rieffel1994quantization, hawkins2008obstruction, landsman2012mathematical}, but they are not quite the same. The lack of submultiplicativity for the norm (\ref{FN}) prevents us from treating $\mathscr{L}(\hn)$ as a $C^\ast$-algebra (or even as a Banach algebra).  Moreover, replacing the normalized Frobenius norm with the operator norm $\|\cdot \|$ would violate the von Neumann condition, as shown in Appendix~\ref{A}.
    \end{remark}

\begin{remark}
Condition ($iii$) implies that the quantization map $\qq^\R$ is asymptotically injective and induces a family of seminorms $|a|_N:=\|\hat{a}\|_N$ on $\fa$ that separates points. It also ensures that, in the classical limit, the root-mean-square values of a physical observable coincide whether evaluated on the classical phase space or on the quantum spectrum. This is a natural and desirable feature from a physical point of view. Since the Hermitian inner product is completely determined by the associated norm, we further obtain
         \begin{equation}
             \lim_{N\to\infty}(\hat a, \hat b)_N=(a,b)_\R\,,\qquad \forall a,b\in \fa\,.
         \end{equation} 
        \end{remark}
\begin{remark}
Compared to strict deformation quantization and related approaches, Dirac's condition ($v$) holds in a weaker sense -- namely, for ``almost all'' states in the sense of dimension.   The ``exceptional'' states are represented by the elements of the quotient space $~{\hh_N/\hh^{a,b}}$, whose dimension becomes negligible compared to $\hh^{a,b}$ in the classical limit $\hbar\rightarrow 0$.          
\end{remark}
\begin{remark}
    By applying the triangle inequality to the normalized Frobenius norm, one finds that the quantum observables asymptotically commute: 
    \begin{equation}\label{zzN}
      \lim_{N\to \infty}\|\hat a\hat b-\hat b\hat a\|_N=0\,.
  \end{equation}
  \end{remark}
        
\begin{remark}
    The periodic boundary condition (\ref{PBC}) imposed on the wave functions allows us to interpret the momentum space as a circle of circumference $M$. Consequently, all formulas above remain valid and lead to a quantization on the bounded cylinder  $S^1\times I$. 
\end{remark}
       
Observe that the von Neumann and Dirac conditions of Proposition \ref{P1} are formulated differently: the former involves the normalized Frobenius norm, whereas the latter uses the usual operator norm and holds for ``almost all'' states. Nevertheless, it is possible to provide a more uniform treatment of these conditions by viewing $\hh_N=P_N\hh$ not as separate Hilbert spaces, but rather as a family of nested subspaces of a single Hilbert space $\hh$.\footnote{The choice of ambient space $\hh$ for the family $\{\hh_N\}$ is not unique. Instead of the Hilbert space $L^2(S^1)$ of square-integrable functions, one could 
consider, for example, the space of Laurent polynomials in $z=e^{\frac{2\pi i p}{M}}$ or the spaces  $C^k(S^1, \mathbb{C})$ of $k$-times differentiable functions on the circle. In all these cases, the elements of $\hh$ are represented by Fourier series 
$
\psi(p)=\sum_{m\in \mathbb{Z}}a_me^{\frac{2\pi imp}{M}}
$
with $a_m\rightarrow 0$ as $|m|\rightarrow \infty$. } 
Within this framework, the following result provides additional insight into how operator products and commutators behave in the classical limit.
\begin{proposition}\label{P2} Let the position space be the interval $[u,v]$ with $u, v\neq 0$. Then,
    for any $a,b\in \fa$ and $\psi \in\hh $, we have the following limits:
    \begin{equation}
        \lim_{N\to\infty}\big\|(\hat{a}\hat{b}-\widehat{ab})P_N\psi \big\|=0\,,\qquad   \lim_{N\to \infty}\big\|\big({(i\hbar)^{-1}}[\hat{a},\hat{b}]-\widehat{\{a,b\}}\big)P_N\psi \big\|=0\,.
        \end{equation}
\end{proposition}
The assumption that the boundary points of $ [u,v]$ are nonzero is not restrictive, since applying an obvious unitary transformation in $\hh$ shifts $\hat x$ to $\hat{x}+d$ for arbitrary $d\in \mathbb{R}$ while leaving the operator $\hat{z}$ intact. The proof of Proposition \ref{P2} is provided in Appendix \ref{A}. 

\begin{example} The first nontrivial case of quantization corresponds to the two-level system with 
$~{\mathcal{H}_2=\mathrm{span}\{ \psi_1,  \psi_2\}\simeq \mathbb{C}^2}$ and
\begin{equation}\label{22}
\begin{array}{c}
    \hat 1=\left(\begin{array}{cc}
        1 &0  \\
        0 & 1
    \end{array}\right)\,,\qquad \widehat{z}=\left(\begin{array}{cc}
        0 &1  \\
        0 & 0
    \end{array}\right)\,,\qquad \widehat{\bar z}=\left(\begin{array}{cc}
        0 &0  \\
        1 & 0
    \end{array}\right)\,,\\[7mm]
\widehat{z^n}=0\,,\qquad n\neq \pm 1\,, \qquad
 \displaystyle   \hat{x}=\frac{2\pi\hbar}{M}\left(\begin{array}{cc}
        1 &0 \\
        0 & 2
    \end{array}\right)\,.
    \end{array}
\end{equation}
The quantization map $\qq^\R: \fa\rightarrow \mathscr{L}(\mathcal{H}_2)$ is surjective and the operators (\ref{22}) span the entire algebra of complex  $2\times 2$ matrices $\mathscr{L}(\hh_2)$.
In contrast to the classical Poisson brackets $\{z,\bar z\}=0$, 
\begin{equation}
    [\widehat z,\widehat{\bar{z}}]=\left(\begin{array}{cc}
        1 &0  \\
        0 & -1
    \end{array}\right)\neq 0\,.
    \end{equation}
    
\end{example}

We see that the compactness of the momentum space has far-reaching physical implications. Most importantly,  the Hilbert space of states  becomes finite-dimensional whenever the particle is confined to a bounded interval. Consequently, the spectrum of the position operator is finite and equidistant: there are a finite number of spatial points where the particle can be observed through measurements.  Below, we consider two more interesting consequences of compactness. 

\subsection{Supertunneling} 
So far we have discussed a quantum particle confined within a connected, bounded interval.  Now, let us consider a configuration space consisting of two disjoint segments, $I_1=[0, x_1)$ and $I_2=(x_2, x_3]$, separated by the intermediate segment $I=[x_1,x_2]$. The corresponding phase space is given by the disjoint union $$I_1\times [0,M]\bigsqcup I_2\times [0,M]\,.$$ Suppose that, initially, the particle is located in  $I_1$ with probability $1$. What, then, is the probability of finding the particle in the segment $I_2$ after some time $t$? 

To simplify the analysis, we assume that the segments $I_1$ and $I_2$ are sufficiently small to support only a single quantum state, i.e., 
$$
x_1=x_3-x_2=\frac{2\pi \hbar}{M}\,.
$$
Choose $n\in \mathbb{N}$ such that $2\pi \hbar n/M\in I_2$.  By construction, we have $n\geq|I|M/2\pi \hbar$. The initial state is then described by the wave function $\psi_0(p)=1$, and the wave function $~{\psi_{n}(p)=e^{\frac{2\pi i n p}{M}}}$ corresponds to the particle being located at $I_2$. These wave functions  span the two-dimensional state space of the system.  The classical Hamiltonian of a free particle is defined by its kinetic energy, which is assumed to be a smooth function of the particle's momentum:
\begin{equation}\label{HH}
     H(p)=\sum_{k\in \mathbb{Z}} H_k z^k\,, \qquad z=e^{\frac{2\pi i p}{M}}\,.
\end{equation}
Since 
\begin{equation}
\|H\|^2_\R=\sum_{k\in \mathbb{Z}}|H_k|^2<\infty\,,
\end{equation}
it follows that $H_k\to 0$ as $k\to \pm \infty$. For the two-state system under consideration, the quantum Hamiltonian is given by the operator
\begin{equation}\label{HHH}
    \hat H= H_0+H_n \widehat{z^n} +\bar H_n \widehat{z^{-n}}\,,
\end{equation}
and the evolution operator takes the form
\begin{equation}
    \hat U(t)=e^{-\frac{it}{\hbar}\hat H}=e^{-\frac{itH_0}{\hbar}}\left(\cos(\omega_n t) \I -i \frac{\sin(\omega_n t)}{|H_n|} (\hat H-H_0)\right)\,, \qquad \omega_n =\frac{|H_n|}{\hbar}\,.
    \end{equation}
The transition probability is given  by the standard formula
\begin{equation}
    P_{0\rightarrow n}(t)=|\langle \psi_n|\hat U(t) |\psi_0\rangle|^2=\sin^2(\omega_n t)\,.
\end{equation}
During time evolution, the system oscillates between the two possible states, with the probability oscillating accordingly. 
If the distance $|I|$ between the segments $I_1$ and $I_2$ remains fixed, then in the decompactification limit $M\rightarrow \infty$, it follows that  $n\rightarrow \infty$ and $\omega_n \rightarrow 0$. As a result, no transition between the two states is possible ($P_{0\rightarrow n}=0$) unless the momentum space is compact. This result is consistent with physical intuition: neither classical nor quantum particles can penetrate an infinitely high potential barrier of finite thickness. Such tunneling, however, becomes possible for a quantum particle with compact momentum space.  This explains why we refer to this phenomenon -- a nonzero transition amplitude -- as {\it supertunneling}. The qualitative picture does not change if we broaden the intervals $I_1$ and $I_2$ to accommodate additional particle states.  This situation bears some resemblance to relativistic quantum field theory. The inability to localize a relativistic particle within a finite region of space ultimately stems from the existence of a maximal velocity, just as the supertunneling phenomenon arises from the upper bound on the particle's momentum.

It is also instructive to compare this result with the transition amplitude between the states $\psi_0$ and $\psi_n$, assuming that the configuration space is the entire real line $\mathbb{R}$, with no gaps. For simplicity, we focus on a particular Hamiltonian of the form (\ref{HH}), namely,  
\begin{equation}
    H=E\left(1- \frac{z+\bar z}2\right)\,,\qquad z=e^{\frac{2\pi i p}{M}}\,,
\end{equation}
where $E$ is  a constant with units of energy. On the entire real line, the evolution operator is expressed as 
\begin{equation}
    \hat U(t)=e^{-\frac{it}{\hbar}\hat H} =e^{-it\omega}\Big[J_0(t\omega)+\sum_{k=1}^\infty i^kJ_k(t\omega)\big( {z^k}+
    {z^{-k}}\big)\Big]\cdot \,,
\end{equation}
where $\omega=E/\hbar$ and $J_k(w)$ are the Bessel functions of the first kind. The probability of transition is given by 
\begin{equation}
    P_{0\rightarrow n}(t)=J^2_n(t\omega)\,.
    \end{equation}
Using the asymptotic form of the Bessel functions, we find 
\begin{equation}
    P_{0\rightarrow n}(t)\sim \frac{2}{\pi t\omega}\cos^2 \left( t\omega-\frac{n\pi}{2}-\frac{\pi}{4} \right)\,,\qquad t\omega \gg n^2\,.
\end{equation}
As before, the probability oscillates, but its amplitude gradually decreases over time, and the oscillation frequency $\omega$ is independent of the spatial separation between $x_1$ and $x_2$. Thus, whenever the momentum space is compact, the particle has a nonzero probability of jumping from the origin $x=0$ to a point in $I_2$, except at certain times. 
To highlight the contrast with conventional quantum mechanics, let us set $E=M^2/4\pi^2 m$, where $m$ is the mass of the particle.  Then 
\begin{equation}\label{FP}
H=\frac{M^2}{4\pi^2m}\Big(1-\cos\Big(\frac{2\pi p}{M}\Big)\Big )=\frac{p^2}{2m}+\mathcal{O}\Big(\frac{1}{M^2}\Big)\,.
\end{equation}
For sufficiently small momenta ($|p|\ll M$), the particle obeys the standard quadratic energy-momentum relation, which becomes exact in the decompactification limit $M\rightarrow \infty$. In this limit, $\omega\rightarrow \infty$ and $P_{0\rightarrow n}\rightarrow 0$. Again, this agrees with the conventional quantum picture: the probability of a particle transitioning from a given point to any finite interval vanishes, since the wave function  becomes spread uniformly along the real line \cite{hegerfeldt1998instantaneous}. 

\subsection{Degenerate Fermi gas}

As another example,  consider the gas of free Fermi particles confined to the interval $(0,L]$ at zero temperature. All particles are assumed to be scalar, noninteracting,  and indistinguishable.  The dynamics of each particle are governed by the Hamiltonian (\ref{FP}). As before, let $N$ denote the dimension of the single-fermion state space $\hn$, and let $ \mathcal{N}$ represent the total number of fermions. In the basis $\big\{\psi_n(p)=z^n\big \}_{n=1}^N$, the quantum Hamiltonian takes the form of a tridiagonal Toeplitz matrix\footnote{Curiously, the matrix is proportional to the Cartan matrix of type $A_N$.}:
\begin{equation}\hat H=\frac E2\left(
\begin{array}{cccccc}
     2&-1&0&\cdots&0  \\
     -1&2& -1&\ddots&\vdots\\
     0&-1&2&\ddots &0\\
     \vdots&\ddots&\ddots&\ddots&-1\\
     0&\cdots &0&-1&2
\end{array}\right)\,,\qquad E=\frac{M^2}{4\pi^2 m}\,.
\end{equation}
The eigenvalues of this matrix determine the energy spectrum of a single fermion. To compute the spectrum, we introduce the $N\times N$ matrix $\hat T$ via
\begin{equation}\label{T}
    \hat H= E\big (\I- \frac12 \hat T\big)
\end{equation}
and set $P_N (\lambda):=(-1)^N\det(\hat{T}-\lambda \I)$. Expanding the last determinant along the first row yields the relation 
\begin{equation}
    P_{N+1}(\lambda )=\lambda P_N(\lambda )-P_{N-1}(\lambda )\,,
\end{equation}
in which one may recognize the recurrence relation for the Chebyshev polynomials of the second kind on the interval $[-2,2]$. An alternative definition is 
\begin{equation}\label{CP}
    P_n(2\cos \phi)=\frac{\sin\, (n+1)\phi}{\sin\,\phi}\,.
\end{equation}
In particular,
\begin{equation}
P_0=1\,,\qquad P_1=\lambda\,,\qquad P_2=\lambda^2-1\,,\qquad P_3=\lambda^3-2\lambda\,,\qquad P_4=\lambda^4-3\lambda^2+1\,.
\end{equation}
The roots of the $n$-th Chebyshev polynomial follow directly from Eq. (\ref{CP}):
\begin{equation}\label{roots}
    \lambda_k=2\cos\, \frac{k\pi}{n+1}\,\qquad k=1,2,\ldots, n\,.
\end{equation}
This leads to the expression 
\begin{equation}
    P_n(z)=\prod_{k=1}^n\Big(z-2\cos\,\frac{k\pi}{n+1}\Big)\,.
\end{equation}
Using Rel. (\ref{T}), it is easy  to express the energy levels of a free  particle through the roots (\ref{roots}) of Chebyshev's polynomials\footnote{An interesting discussion of Cartan matrices and Chebyshev polynomials can be found in \cite{damianou2011characteristic}.}:
\begin{equation}
    E_k=E\big(1-\cos k\theta_N\big)\,, \qquad \theta_N=\frac{\pi}{N+1}\,,\qquad k=1,2,\ldots, N\,.
\end{equation}

The Pauli exclusion principle sets an upper bound on the number of identical fermions that can be accommodated within a segment $(0, L]$, namely 
$\mathcal{N}\leq N=LM/2\pi \hbar$.  This gives rise to the {\it maximal density} $n_{\mathrm{max}}$: if $n=\mathcal{N}/L$ denotes the density of the 1D Fermi gas, then  
\begin{equation}
    n\leq n_{\mathrm{max}}=\frac{N}{L}=\frac{M}{2\pi \hbar}\,,\qquad \mbox{and}\qquad \frac{n}{n_{\mathrm{max}}}=\frac{\mathcal{N}}{N}\leq 1\,.
\end{equation}
The total energy of the degenerate Fermi gas is given by the sum 
\begin{equation}\label{Uex}
\begin{array}{c}
  \displaystyle  U(\mathcal{N})=\sum_{k=1}^{\mathcal{N}}E_k=E\left(\mathcal{N}+1-\mathrm{Re}\sum_{k=0}^{\mathcal{N}} e^{ik\theta_N}\right)=E\left(\mathcal{N}+1-\mathrm{Re}\frac{1-e^{i(\mathcal{N}+1)\theta_N}}{1-e^{i\theta_N}}\right)\\[7mm]
  \displaystyle  =E(\mathcal{N}+1)-\frac E2\left(1+\frac{\cos {\mathcal{N}\theta_N} -\cos(\mathcal{N}+1)\theta_N}{1-\cos{\theta_N}}\right)\\[7mm]\displaystyle =E\left(\mathcal{N}+\frac12\right) -\frac E2\left( \cos{\mathcal{N}\theta_N}+\sin\mathcal{N}\theta_N\,\mathrm{ctg}\left(\frac{\theta_N}{2}\right)\right)\,.   
    \end{array}
\end{equation}
In particular, the maximum energy that can be stored in the Fermi gas is equal to   $$U_{\mathrm{max}}=U({N})=E{N}\,.$$ For the case
$N\gg \mathcal{N} \gg 1$, the energy can be approximated as
\begin{equation}
    U(\mathcal{N})\approx E\mathcal{N}+\frac E4\left(\frac{\mathcal{N}}{N}\right)^2-\frac{EN}{\pi}\sin\frac{\pi \mathcal{N}}{N} \approx \frac{\pi^2}6 \frac{E\mathcal{N}^3}{N^2}=\frac{\pi^2}{6}\frac{\hbar^2\mathcal{N}^3}{m L^2}\,,
\end{equation}
from which the degeneracy pressure of the dilute fermion gas follows as
\begin{equation}\label{n3}
    P=-\frac{\partial U}{\partial L}\approx\frac{\pi^2\hbar^2}{3m} n^3\,. 
\end{equation}
Thus, the pressure scales as the cube of density, as expected for the 1D Fermi gas with an approximately quadratic dispersion relation (\ref{FP}). In the general case, $P$ is a function of both the concentration $n$ and the number of available states $N$ (or equivalently, the system size $L$). In the thermodynamic limit $N\rightarrow \infty$, however, the explicit dependence on $L$ disappears, and the pressure reaches its maximum at the highest possible concentration (see Eq. (\ref{Pres}) below).  

It is instructive  to compare these results with semiclassical computations based on the Hamiltonian (\ref{FP}).  In the semiclassical approximation, the total energy of the Fermi gas is given by the integral
\begin{equation}\label{U}
\begin{array}{c}
  \displaystyle  U=\frac{1}{2\pi\hbar}\int^{p_F}_{-p_F}H(p)dpdx=\frac{L}{2\pi\hbar}\int_{-p_F}^{p_F}\frac{M^2}{4\pi^2m}\left(1-\cos\left(\frac {2\pi p}{M}\right)\right)dp\\[7mm]
  \displaystyle  =\frac{M^2L}{4\pi^3\hbar m}
\left( p_F-\frac{M}{2\pi}\sin\left(\frac{2\pi p_F}{M}\right)\right)  \,,
\end{array}
\end{equation}
where the Fermi momentum $p_F$ is determined by the Bohr--Sommerfeld condition:
\begin{equation}
    \mathcal{N}=\frac{1}{2\pi\hbar}\int_{-p_F}^{p_F}dpdx=\frac{p_FL}{\pi\hbar}\,.
\end{equation}
Substituting this into (\ref{U}), we get
\begin{equation}
    U=\frac{M^2\mathcal{N}}{4\pi^2 m}
 -\frac{LM^3}{8\pi^4\hbar m}\sin\left(\frac{\pi \mathcal{N}}{N}\right)\,, 
 \end{equation}
 from which the degeneracy pressure follows:
\begin{equation}\label{Pres}
\begin{array}{rcl}
    P=\displaystyle -\frac{\partial U}{\partial L}&=&\displaystyle \frac{\partial}{\partial L}\left[\frac{LM^3}{8\pi^4\hbar m}\sin\left(\frac{2\pi^2\hbar\mathcal{N}}{LM}\right)\right]\\[5mm]
   &=&\displaystyle \frac{M^3}{8\pi^4\hbar m}\left[\sin\Big(\frac{\pi\mathcal N}{N}\Big) -\Big(\frac{\pi\mathcal N}{N}\Big) \cos\Big(\frac{\pi\mathcal N}{N}\Big) \right]\,.
    \end{array}
\end{equation}
Setting $\mathcal{N}=N$ yields the maximum pressure $P_{\mathrm{max}}=M^3/8\pi^3\hbar m$. For $\mathcal{N}\ll N$ or in the decompactification limit $M\to\infty$,  we recover the cubic law (\ref{n3}). 

\section{More general symplectic domains}

A bounded cylinder $S^1\times I$, discussed in Sec. \ref{RecDom}, may be viewed as a bounded domain in the cotangent bundle of the abelian group $U(1)$. Using elementary harmonic analysis on $U(1)$, we constructed a quantization of the phase space $S^1\times I$, leading to a family of finite-dimensional algebras of Hermitian operators. In this section,  we generalize this quantization technique to bounded domains associated with arbitrary compact Lie groups.

\subsection{Classical observables} 

Let $T^\ast G\simeq G\times \gs$ be the cotangent bundle of a compact Lie group $G$ and let $D$ be an open domain in $\gs$. If $\dim G=d$, then $G\times D\subset T^\ast G$ is an open submanifold of dimension $2d$ with symplectic structure obtained by restricting the canonical symplectic form on $T^\ast G$.  We regard the direct product $M=G\times D$ as the phase space of a classical mechanical system, with $D$ interpreted as the position space and  $G$ as the momentum space. The corresponding Poisson brackets are defined by
\begin{equation}
    \{f, h\}=0\,,\qquad \{f, x_a\}=\ell_af\,,\qquad \{x_a,x_b\}=-f_{ab}^cx_c\,.
\end{equation}
Here $f, h\in C^\infty(G)$, $x_a$ are Cartesian coordinates on $\gs$, and  $\ell_a$ form the dual basis of left-invariant vector fields on $G$, satisfying the commutation relations $[\ell_a,\ell_b]=f_{ab}^c\ell_c$. In the case where  $G$ is semisimple, we choose the basis $\ell_a$ in such a way that the Cartan--Killing metric $g_{ab}=f^c_{ad}f^d_{bc}$ is given by the negative Kronecker's $-\delta_{ab}$.

Denote by $\fg$ the algebra of representative functions on $G$. By definition,  this algebra is spanned by the matrix elements of unitary irreducible representations of $G$. If $\hat G$ denotes the set of all equivalence classes of irreducible representations of $G$, then a generic element of $\fg$ can be written as
\begin{equation}
    f(g)=\sum_{[\pi]\in \hat G}f_\pi^{ij}\pi_{ij}(g)\,,
\end{equation}
where $\big(\pi_{ij}(g)\big)$ are the matrix elements of a representation $\pi: G\rightarrow GL(V)$ and only finitely many coefficients $f_\pi^{ij}\in \mathbb{C}$ are nonzero. Let $\pol$ denote the algebra of polynomial functions on $\gs$, and $\pold$ its restriction to the domain $D\subset \gs$. 
The Poisson algebra of smooth complex-valued functions on $M$ contains a subalgebra $\fa=\fg\otimes \pold$, whose elements are given by finite linear combinations of monomials 
\begin{equation}
    f(g)x_{a_1}\cdots x_{a_n} \,,\qquad n\in \mathbb{N}_0\,,\quad f\in \fg\,.
    \end{equation}
We identify $\fa$ with the {\it Poisson algebra of classical observables}. 

By the Stone--Weierstrass theorem, $\fa$ is dense in the space $C^\infty(M, \mathbb{C})$ with respect to the supremum norm $|\cdot|_\infty$. 
For our subsequent calculations, it will be more convenient to introduce another norm on
$\fa$, which we call the {\it 
$R$-norm}. It is defined as follows. Using multi-index notation, we can write each element of $\fa$ as
\begin{equation}\label{MI}
    a=\sum_{\bf{n}}a_{\bf n}(g)x^{\bf{n}}\,,
\end{equation}
where ${\bf n}=(n_1,\ldots,n_d)$ and $x^{\bf n}:=x_1^{n_1}\cdots x^{n_d}_d$. Given $R>0$, we define 
\begin{equation}\label{RN}
|a|_R=\sum_{\mathbf{n}}|a_{\mathbf{n}}|_{\infty}R^{|\mathbf{n}|}\,,
\end{equation}
where $|\mathbf{n}|:=n_1+\cdots+n_d$. The $R$-norm satisfies all the necessary properties (positivity, homogeneity, subadditivity, and submultiplicativity), thereby turning $\fa$ into a normed commutative algebra.   Suppose the domain $D\subset \gs$ is bounded and entirely  contained in the ball of radius $R$ centered at $0\in \gs$. The minimal $R$ with this property will be called the {\it radius of the domain $D$}. Clearly, for any such domain, $|a|_\infty\leq |a|_R$. 

In addition to the norm topology, the algebra $\fa$ enjoys a bifiltration determined by the natural filtrations of $\pold$ and $\fg$. The filtration 
in $\pold$ is the standard  filtration by monomial degree:
\begin{equation}
    \mathrm{Pol}^n(D)\subset \mathrm{Pol}^{n+1}(D)\,,\qquad \mathrm{Pol}^n(D) \, \mathrm{Pol}^m(D)\subset \mathrm{Pol}^{n+m}(D)\,,
    \end{equation}
where the subspace $\mathrm{Pol}^n(D)$ is spanned by all polynomials of  degree less than or equal to $n$. 

To introduce the second filtration, we need some additional terminology. Recall that all finite-dimensional (not necessarily irreducible) representations of a compact Lie group $G$ (considered up to equivalence) form a semiring $\mathfrak{R}(G)$, with direct sum and tensor product as addition and multiplication, respectively.
If $G$ is simply connected (and hence semisimple), then each representation $\pi\in \mathfrak{R}(G)$ both defines and is defined by a finite-dimensional representation $\dot \pi$ of the corresponding Lie algebra $\mathfrak{g}$. The latter, in turn, is completely determined by its character:
\begin{equation}
    \mathrm{ch}(\dot \pi)=\sum_{\lambda\in P} m_\lambda e^{\lambda}\,.
\end{equation}
Here $P=P(\dot \pi)$ is the system of weights of the representation $\dot \pi$, and $m_\lambda$ is the multiplicity of the weight $\lambda\in P$. Recall that the weights of $\mathfrak{g}$ form a lattice $\Lambda\subset \mathfrak{h}^\ast$ in the dual space to the Cartan subalgebra $\mathfrak{h}\subset \mathfrak{g}$, and that the dominant weights constitute a sublattice $\Lambda^+\subset \Lambda$. The additive and multiplicative properties of the character,
\begin{equation}
    \mathrm{ch}(\dot \pi_1\oplus \dot \pi_2)=\mathrm{ch}(\dot \pi_1)+\mathrm{ch}(\dot \pi_2)\,,\qquad   
    \mathrm{ch}(\dot \pi_1\otimes  \dot \pi_2)=\mathrm{ch}(\dot \pi_1)\cdot \mathrm{ch}(\dot \pi_2)\,,
\end{equation}
allow us to regard the character as a homomorphism from the semiring $\mathfrak{R}(G)$ to the group ring $\mathbb{Z}[\Lambda]$ of the free abelian group $\Lambda$. 
Define the {\it Casimir length} $|\pi|\geq 0$ of a representation $\pi\in \mathfrak{R}(G)$ by the formula
\begin{equation}\label{cl}
    |\pi|^2= \max_{\lambda\in P(\dot\pi)} \langle \lambda,\lambda+2\rho\rangle\,.
\end{equation}
Here, the angle brackets denote the inner product in $\mathfrak{h}^\ast$ induced by the Killing form, and $\rho$ is the sum of the fundamental weights (the generators of the lattice $\Lambda^+$), often called the Weyl vector.  It is known that the set of weights $P(\dot\pi)$ is invariant under the action of the Weyl group $W$ and that each weight of $\Lambda$ is conjugate to a dominant weight from $\Lambda^+$.  Since $\langle \lambda,\rho\rangle>0$ for any $\lambda\in \Lambda^+$, the maximum in (\ref{cl}) is attained on  dominant weights. By restricting to the dominant weights, one can easily see that
\begin{equation}
    |\pi_1\otimes \pi_2|\leq |\pi_1|+|\pi_2|\,.
\end{equation}
This induces a filtration on the semiring $\mathfrak{R}(G)$ defined by the family of additive subgroups
\begin{equation}\label{t}
    \mathfrak{R}_l(G)=\big\{\pi \in \mathfrak{R}(G)\; \big|\;  |\pi|\leq l\;\big \}\,,\qquad l\in \mathbb{R}_{\geq 0}\,.
\end{equation}
By definition, 
$$
    \mathfrak{R}_{l_1}(G)\subset \mathfrak{R}_{l_2}(G)\quad \mbox{whenever}\quad l_1\leq l_2\,,
    $$
    and
    $$
    \mathfrak{R}_{l_1}(G)\,\mathfrak{R}_{l_2}(G)\subset \mathfrak{R}_{l_1+l_2}(G) \,,\qquad \mathfrak{R}(G)=\bigcup_{l \geq 0}\mathfrak{R}_l(G)\,.
    $$
Thus, the family $\mathfrak{R}_l(G)$ defines an exhaustive multiplicative filtration on $\mathfrak{R}(G)$, indexed by the abelian semigroup $(\mathbb{R}_{\geq 0}, +)$. 

This filtration induces a multiplicative filtration on the algebra of representative functions $\fg$. Define the subspace of {\it Casimir length at most} $l$ as 
\begin{equation}\label{FF}
    \mathcal{F}_l(G)=\mathrm{span} \big\{\pi_{ij}(g)\in \fg\;\big|\; \pi\in \mathfrak R_l(G)\;\big\}\,.
\end{equation}
Then 
$$
 \mathcal{F}_{l_1}(G) \,\mathcal{F}_{l_2}(G)\subset  \mathcal{F}_{l_1+l_2}(G)\,, \qquad \bigcup_{l\geq 0}\mathcal{F}_l(G)=\fg\,.
 $$
Associated with the quadratic Casimir element $C_2=-\delta^{ab}t_a t_b$ is the Laplace operator \begin{equation}
\mathbb{L}=-\delta^{ab}\ell_a\ell_b\end{equation}
on scalar functions.  It is well known that the eigenfunctions of $\mathbb{L}$ are precisely  the matrix elements of irreducible representations of $G$. Furthermore,
\begin{equation}
    \mathbb{L}\, \pi_{ij}=|\pi|^2\pi_{ij}\,,\qquad \forall\, [\pi] \in \hat G\,.
\end{equation}
This explains the terminology ``Casimir length'' for $|\pi|$. Each eigenvalue $|\pi|^2$ has finite multiplicity, and the classical Weyl law states that 
 \begin{equation}\label{WR}
     \dim \mathcal{F}_l(G)=\frac{\mathrm{Vol}(G) }{(4\pi)^{d/2}\Gamma(\frac d2+1)} l^d+\mathcal{O}(l^{d-1})\,,
 \end{equation}
where $d=\dim G$ (see \cite{Minakshisundaram_Pleijel_1949}). Thus, Weyl’s law shows that the filtration grows polynomially in $l$, which will provide effective control over the size and complexity of $\fg$ in the process of quantization.

Let us denote 
$$
\fa^n_l=\mathcal{F}_l(G)\otimes \mathrm{Pol}^n(D)\,,\qquad \fa^n=\fg\otimes \mathrm{Pol}^n(D)\,,\qquad \fa_l=\mathcal{F}_l(G)\otimes \mathrm{Pol}(D)\,.
$$
Combining both filtrations, we introduce  the {\it total filtration}: 
\begin{equation}\label{TF}
\fa^{(t)}=\bigoplus_{n+l=t}\fa_l^n\,,\qquad \fa^{(t)}\subset \fa^{(t')}\quad\mbox{whenever}\quad t\leq t'\,.
\end{equation}
This total filtration will be especially useful in our subsequent constructions. 

\subsection{Prequantization}\label{pre-Q}
Consider the pre-Hilbert space $\hh=\fg$ of representative functions equipped with the Hermitian inner product
\begin{equation}
    (f,h)=\int_G\overline{ f(g)} h(g)dg\qquad \forall f, h\in \hh\,,
\end{equation}
where $dg$ denotes the normalized Haar measure on $G$. By the Peter--Weyl theorem, $\hh$ is dense in the Hilbert space $L^2(G, dg)$. Let $\mathscr{L}(\hh)$ denote the algebra of operators on $\hh$.  For any real number $\hbar\geq 0$, define a $\mathbb{C}$-linear  map 
$\qq:\fa\rightarrow \mathscr{L}(\hh)$ by the relations:
\begin{equation}\label{PQ}
\begin{array}{c}
   \qq(f)\psi=f\cdot \psi\,,\qquad  \qq(x_a)\psi=-i\hbar\ell_a\psi \,,\qquad \forall f\in \fg\,,\quad\forall \psi\in\hh\,,\\[5mm]
 \displaystyle    \qq(fx_{a_1}\cdots x_{a_n})=\frac{1}{(n+1)!}\sum_{k=0}^n\sum_{\sigma\in S_{n}} \sigma\big( \qq(x_{a_1})\cdots \qq(x_{a_k})\qq(f)\qq(x_{a_{k+1}})\cdots \qq(x_{a_n})\big)\,.
 \end{array}
\end{equation}
where the inner summation runs over all permutations of the symbols $x_{a_1},\ldots, x_{a_n}$.  We will refer to $\qq$ as the {\it prequantization map}. It follows immediately from the definition that
\begin{equation}\label{QQ}
   [\qq(x_a), \qq(x_b)]=-i\hbar f_{ab}^c \qq(x_c),  \quad [\qq(x_a),\qq(f)]=-i\hbar\qq(\ell_af)\,,\ \quad [\qq(f),\qq(h)]=0
\end{equation}
 for all $f,h\in \fg$.  The first commutator shows that the operators $\qq(x_a)$ form a Lie algebra which is isomorphic to $\mathfrak{g}$ for $\hbar>0$. Then the second and third commutators  generate a trivial  extension  of  $\mathfrak{g}$ by the $\mathfrak{g}$-module $\fg$. Observe that the action of $\mathfrak{g}$ on $\fg$ respects the filtration by the Casimir length, so that we can speak of an infinite-dimensional  Lie algebra $~{\mathcal{L}=\mathfrak{g}\subset\hspace{-1em}+\fg}$ filtered by the family of finite-dimensional subalgebras $~{\mathcal{L}_l=\mathfrak{g}\subset\hspace{-1em}+\mathcal{F}_l(G)}$. 
The prequantization map $\qq$ therefore defines a unitary representation of the universal enveloping algebra $\mathcal{U}(\mathcal{L})$ in the pre-Hilbert space $\hh$. This allows us to identify the space of prequantum observables $\qq(\fa)$ with the quotient algebra $\mathcal{U}(\mathcal{L})/\mathcal{J}$, where the two-sided ideal $\mathcal{J}$ is generated by the relations 
 $$
 fh-f\otimes h=0\,,\qquad \forall f,h\in \fg\,.
 $$

Another possible interpretation of the commutation  relations (\ref{QQ}) is provided by the notion of Lie--Rinehart pair 
\cite{LR}. The Poisson algebra of classical observables $\fa$ contains the subspace $\fa^0\oplus\fa^1$ consisting of functions that are at most linear in $x_a$. The latter carries a natural structure of the Lie--Rinehart pair $(\mathfrak{M}, \mathfrak{L})$, where the (infinite-dimensional) Lie algebra $\mathfrak{L}$ is given by $\mathfrak{A}^1$, and the $\mathfrak{L}$-module $\mathfrak{M}$ is identified with $\mathfrak{A}^0$. Both the Lie bracket on $\mathfrak{L}$ and the (left) action of $\mathfrak{L}$ on $\mathfrak{M}$ are defined by the Poisson bracket, while the left action of $\mathfrak{M}$ on $\mathfrak{L}$ is given by the dot  product  of functions. In geometrical terms, this Lie--Rinehart pair corresponds to the tangent Lie algebroid of $G$. 
With this interpretation,  the quantization map $\qq$ defines a homomorphism from the universal enveloping algebra $\mathcal{U}(\mathfrak{M},\mathfrak{L})$ of the Lie--Rinehart pair to the algebra of differential operators on $G$.

 The next statement, for which the proof is left to the reader, represents a form of the correspondence principle. 
\begin{proposition} \label{P3} The following statements hold:
\begin{enumerate}
 \item The prequantization map is Hermitian, meaning that $\qq(\bar a)=\big(\qq(a)\big)^\dagger$.
    \item The prequantization map is injective, i.e., $\ker(\qq)=0$.
    \item The image of the prequantization map, $\qq(\fa) $, is a subalgebra of $\mathscr{L}(\hh)$.   
    \item For any $a\in \fa^n$ and $b\in\fa^m$, the following $\ast$-product is defined:
\begin{equation}\label{star}
    a\ast b:=\qq^{-1}\big(  \qq(a)\qq(b)\big)= \sum_{k=0}^{n+m}  \hbar^k D_k(a,b)\,,
\end{equation}
       where 
       $$
        D_0(a,b)=ab\,,\qquad D_1(a,b)=\frac i2\{a,b\}\,,\qquad\mbox{and}\qquad    D_k(a,b)\in\fa^{n+m-k}\,.  
        $$
   If $a$ and $b$ do not depend on $\hbar$, then neither do the elements $D_k(a,b)$. 
   \item The $\ast$-product respects the total filtration (\ref{TF}), that is,
   $
       \fa^{(t)}\ast \fa^{(s)}\subset \fa^{(t+s)}
   $.
   \end{enumerate}
   Thus, the star-product (\ref{star}) defines a deformation quantization of the dense Poisson subalgebra $\fa\subset C^\infty(M,\cc)$.  \end{proposition}

Deformation quantization on the cotangent bundles of Lie groups has been studied by many authors; see, for example, \cite{gutt1983explicit, tounsi2003integral, domanski2021deformation}. 

Using the commutation relations (\ref{QQ}) one can bring prequantum observables (\ref{PQ}) into the form where all operators $\qq(x_a)$ appear to the right of operators $\qq(f)$.  In the multi-index notation (\ref{MI}), this reordering of factors takes the form 
\begin{equation}
    \qq\big(a_{\bf n} x^{\bf n}\big)=\sum_{  |{\bf m}|\leq |{\bf n}|   }\qq\big(a^\hbar_{\bf m}\big)\qq\big(x^{\bf m}\big)
\end{equation}
for some $a^\hbar_{\bf m}\in \fg$. 
 
The operators $\qq(a)$ are generally unbounded unless $a\in \fg$. A more detailed estimate can be obtained as follows. The filtration by the Casimir length (\ref{FF}) extends naturally to the pre-Hilbert space $\hh=\bigcup_{l\geq0} \hh^l$. For $\psi\in \mathcal{H}^l$, we have
\begin{equation}\label{qx}
\|\qq(x_a)\psi\|^2\leq \sum_{b=1}^{d}\|\qq(x_b)\psi\|^2=\delta^{ab}(\qq(x_a)\psi,\qq(x_b)\psi)=
\hbar^2(\psi, \mathbb{L}\psi)\leq \hbar^2l^2\|\psi\|^2\,.
\end{equation}
If $f\in \fg$, then 
\begin{equation}
\|\qq(f)\psi\|^2=\|f\psi\|^2=\int_G |f(g)|^2 |\psi(g)|^2dg\leq |f|_\infty ^2\|\psi\|^2\,.
\end{equation}
Using the definition of $R$-norm (\ref{RN}) with $R=\hbar l$, we finally obtain
\begin{equation}\label{QR}
\|\qq(a)\psi \|\leq |a|_{R} \cdot\|\psi\|\,,\qquad \forall a\in \fa\,, \quad \forall \psi\in \hh^l\,.
\end{equation}
This shows that when restricted to the subspace $\hh^l$, the prequantization map is controlled by the $R$-norm, providing a concrete estimate of the operator size.

\subsection{Quantizable semialgebraic domains}

We consider open domains $D\subset \gs$ that are defined by a finite set of inequalities
\begin{equation}\label{obs}
    d_A(x) > 0\,,
\end{equation}
where $d_A$ are polynomials. Such domains are called {\it semialgebraic} or, more  precisely, {\it basic (open) semialgebraic sets} \cite{Berr2001PositivePO, book:92008691,scheiderer2024course,book:100944507}. To indicate the set of polynomials used to define a basic semialgebraic  set $D$, we write $D=D(d_1, \ldots, d_n)$. Replacing the strict equality $d_A>0$ with the non-strict inequality $d_A\geq 0$ defines a {\it basic closed semialgebraic set}. More general  semialgebraic sets are obtained by applying the set-theoretic operations of intersection, union, and complement. For example, if $D(d_1,\ldots, d_n)$  and $D'(d'_1,\ldots, d_m')$, then 
$$
(D\cap D')(d_1,\ldots,d_n,d'_1,\ldots,d_m')\,.
$$
Clearly, any finite union of open intervals on $\mathbb{R}$ is semialgebraic . The open ball $B^d_R$: 
$$
R^2-(x_1^2+\cdots +x_d^2)>0
$$
is a semialgebraic  set in $\mathbb{R}^d$. We say that a semialgebraic  set $D\subset \mathbb{R}^d$ is {\it bounded} if there exists a ball $B^d_R$ such that $D\cap B^d_R=D$. The minimal $R$ satisfying this property is called the {\it radius of the domain} $D$.

    According to the definition, a basic semialgebraic  domain $D=D(d_1,\ldots,d_n)$ is not just an open set in $\gs$ -- it is an open set defined by a specific collection of polynomials $d_1,\ldots, d_n$. If two different collections of polynomials define the same domain $D\subset \gs$, we regard them as defining different semialgebraic   sets. For example, the inequalities 
    $$
    d=1-x^2-y^2>0\quad \mbox{and}\quad d'=(x^2+3y^2+1)(1-x^2-y^2)>0
    $$
both define the unit disc in the plane, but $D(d)\neq D(d')$.  Given a basic semialgebraic  set $D(d_1,\ldots, d_n)$, one can introduce the following  semirings\footnote{A semiring is a generalization of a ring in which additive inverses are not required.}:
    \begin{itemize}
    \item $S(D)=\{\;f\in \pol\;| \; f|_{\bar D}>0\; \}$ -- the semiring of polynomials strictly positive on the topological closure $\bar D$, defined by the inequalities $d_A(x)\geq 0$; 
        \item $\langle d_1,\ldots, d_n, \mathbb{R}_+ \rangle$ -- the semiring generated by the defining polynomials $d_1,\ldots,d_n$ together with positive reals;  
        \item $\langle d_1,\ldots, d_n, \pol^2 \rangle$ -- the semiring generated by the $d$'s and the squares of polynomials. 
        \end{itemize}
    It is clear that $$\langle d_1,\ldots, d_n, \mathbb{R}_+ \rangle\subset\langle d_1,\ldots, d_n, \pol^2 \rangle\subset S(D)\,.$$
    In general, both inclusions are proper. Consequently,  neither the polynomials $d_1,\ldots,d_n$ nor the semiring they generate are uniquely determined by an open set $D$. However, it is known that 
    $$
    S(B^d_R)=\langle R^2-x_1^2-\cdots-x^2_{d}\,, \,\mathrm{Pol}(\mathbb{R}^d)^2 \, \rangle
    $$
    and
    $$
    S([a,b]^d)=\langle b-x_1,\ldots, b-x_d,x_1-a,\ldots, x_d-a, \mathbb{R}_+\rangle \,.$$
    
Let us also mention the following general result, whose proof can be found in    
\cite{Berr2001PositivePO, book:92008691,scheiderer2024course,book:100944507}.

\begin{theorem}[Schm\"udgen]  \label{T1}
    Let $D=D(d_1,\ldots, d_n)$ be a bounded semialgebraic set. Then $$ S(D)=\langle d_1,\ldots, d_n, \pol^2 \rangle\,.$$ 
\end{theorem}
In other words, any strictly positive polynomial $f$ on a {\it compact} semialgebraic set defined  by inequalities $d_A(x)\geq 0$ can be represented as 
\begin{equation}
 f=\sum_{e_k\in \{0,1\}} f^2_{e_1e_2\cdots e_n}d_1^{e_1}d_2^{e_2}\cdots d_n^{e_n}
\end{equation}
for some polynomials $f_{e_1e_2\cdots e_n}(x)$.

Given  a basic semialgebraic domain  $D(d_1,\ldots,d_n)$, we can treat the polynomials $d_A(x)$ as classical observables from the subalgebra $\fa_0\subset\fa$. When passing to quantization, we allow the polynomials $d_A(x)$ to depend on the parameter $\hbar$. A domain  $D(d_1,\ldots,d_n)$ is said to be {\it quantizable} if 
$$
[\qq(d_A),\qq(d_B)]=0\,,\qquad \forall \hbar\geq 0\,, \quad\forall A,B=1,\ldots,n\,.
$$
Two quantizable domains $D(d_1,\ldots,d_n)$ and $D(d'_1,\ldots,d'_m)$ are called {\it compatible} if their intersection $D\cap D'$ is quantizable\footnote{The case $D\cap D'=\varnothing$ is not excluded.}, meaning that $[\qq(d_A),\qq(d'_B)]=0$. More generally, a semialgebraic  domain is quantizable if it is a union of quantizable basic domains that are mutually compatible.  Hence, the operations of intersection and union are only partially defined for quantizable domains. However, the ball $B^d_R$ associated with the quadratic Casimir element $C_2=t_1^2+\ldots+t_d^2$ is compatible with any quantizable domain, and the same is true for the semialgebraic domains associated with higher Casimir elements. This property allows us to apply the above notions of boundedness and radius to quantizable semialgebraic domains. In the following, we will mostly consider quantizable bounded domains of the form $D(d_A, R^2-C_2)=D(d_A)\cap B^d_R$.

\begin{example}Let us consider the three-dimensional Lie algebra $\mathfrak{g}=\mathfrak{su}(2)$ with commutation relations:
$$
[t_1,t_2]=t_3\,,\qquad[t_2,t_3]=t_1\,,\qquad[t_3,t_1]=t_2\,.
$$
The quadratic Casimir element $t_1^2+t_2^2+t_3^2$ gives rise to a family of open balls of radius $R$ in the dual space $\gs\simeq \mathbb{R}^3$:
$$
B^3_R=D(R^2-x_1^2-x_2^2-x_3^2)\,.
$$
Further examples can be constructed by taking  intersections $B^3_R\cap D$ of the $3$-ball  with a quantizable semialgebraic domain $D$. 
Here are some examples:
\begin{equation}
\begin{array}{c}
    D_1=D(x_3-z)\,,\qquad D_2=D(x_3^2-x_1^2-x_2^2\,,\,-x_3)\,,\\[3mm]
    D_3=D(r^2-x_1^2-x_2^2\,, \,z_1-x_3\,,\,x_3-z_2)\,, \qquad D_4=D(x_1^2+x_2^2+x_3^2-r^2)\,.
    \end{array}
\end{equation}
Geometrically,  the intersection $B^3_R\cap D_1$ represents a spherical cap (provided $R>z>0$), while $B^3_R\cap D_2$ defines a spherical cone (see Fig.\ref{Pic1}). The domain $B^3_R\cap D_3$ is given by the intersection of the ball with a cylinder of radius $r$ and height $z_1-z_2$; for $r>R$, it degenerates to a spherical segment.  Finally, the intersection $B_R^3\cap D_4$ is a spherical shell of thickness $R-r$.    All these bounded domains are quantizable.
\end{example}
\begin{example}
    For the same Lie algebra $\mathfrak{su}(2)$, we may define a finite set of balls 
    $$
    D_k=D\big(R^2_k-x_1^2-x_2^2-(x_3+z_k)^2\big)\,.
    $$
The semialgebraic domain $D=\bigcup_kD_k$ is obviously bounded and quantizable. Depending on the parameters $R_k$ and $z_k$, it may consist of several connected components.    
    
\end{example}

\begin{figure}
    \centering
\begin{tikzpicture}[scale=1.2]
\draw[fill=blue!15] (0, 0) circle (2);
\draw[fill=green!15] (0, 0) circle (1.7);
\draw[fill=red!20, opacity=0.8] (45:2) arc(45:135:2) -- cycle;
\draw[fill=yellow!25, opacity=0.8] (0,0)--(-45:2) arc (-45:-135:2)--cycle;
\draw[->] (0,-2.5) -- (0,2.5);
\draw[->] (-2.5,0) -- (2.5, 0);
\draw (2.5, 0) node [below] {${}_{x_1}$};
\draw (0, 2.4) node [left] {${}_{x_3}$};
\draw[-, fill=blue!15] (4.9, 1.1) -- (9.1, 1.1)--(9.5,1.5)--(4.5, 1.5)--cycle;
\draw [-] (7,-1)--(10,2);
\draw [-] (7,-1)--(4,2);

\draw[-, fill=green!20] (7,-1)--(9.1,1.1)--(4.9 ,1.1)--cycle;
\draw[-, fill=red!20, opacity=0.8]  (9.5,1.5)--(8.6,0.6)--(8.6,1.5)--cycle;

\draw[-, fill=yellow!25, opacity=0.8] (7, -1) .. controls(6.9,-1) .. (5.35, 1.5)--(4.5, 1.5)--cycle;

\draw[-, thin] (7.2,-1.8)--(5,2);
\draw[->] (7,-2.5) -- (7,2.5);
\draw[->] (4,-1) -- (10, -1);

\foreach \y in {0,...,14}
\foreach \x in {0,...,\y}
  \draw[fill=black] (\x/2.5+7- \y/5,\y/5-1) circle (0.5pt);

\draw (7, 2.4) node [left] {${}_{j}$};
\draw (9.9, -1) node [below] {${}_{m}$};
\end{tikzpicture}
 \caption{Left panel:  projection of quantizable semialgebraic domains in $\mathfrak{su}(2)^\ast\simeq\mathbb{R}^3$ onto the plane $x_1x_3$, 
 showing a spherical cap, a cone, and a shell. Right panel: eigenstates of these domains on the weight diagram of the right regular representation of $SU(2)$, where $j\in \frac12\mathbb{N}_0$ and each weight occurs with multiplicity $2j+1$.}
    \label{Pic1}
\end{figure}
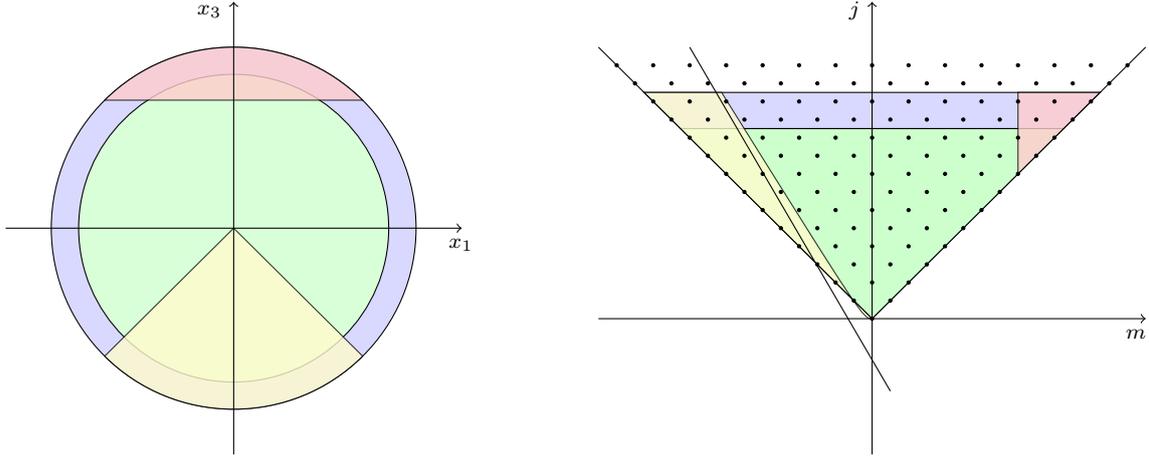

\begin{example} Continuing with $\mathfrak{su}(2)$, consider  the open cube 
$$
D(1-x_1^2\,,\, 1-x_2^2\,,\,1-x_3^2)\,.
$$
Although this semialgebraic domain is not quantizable, we can approximate it arbitrarily closely by the family of solid superellipsoids
$$
D_n=D(1-x_1^{2n}-x_2^{2n}-x_3^{2n}) \,.
$$
It is clear that  each superellipsoid domain is bounded and quantizable, and that $D_n\rightarrow D$ as $n\rightarrow\infty$.
\end{example}

\subsection{Quantization}

We continue with the notation and definitions of Sec. \ref{pre-Q}. Suppose we are given a symplectic domain 
$G\times D\subset T^\ast G$, where  $D=D(d_1,\ldots, d_n)$ is a bounded and quantizable semialgebraic set in $\gs$ of radius $R$.  As the first step, we associate to $D$ a subspace $\hh^D$ of the pre-Hilbert space $\hh=\fg$. This is defined as follows. Using the prequantization map (\ref{PQ}), we first define the operators $\qq(d_A)$. Among these operators we have $R^2 -\hbar^2\mathbb{L}$. 
By definition, the operators $\qq(d_A)$ are Hermitian and commute with each other. Therefore, we can speak of their common eigenfunctions defined by the equations:
$$
\qq(d_A)\psi=\lambda_A\psi\,,\qquad A=1,\ldots, n\,.
$$
The eigenfunctions that correspond to the positive spectrum of the $\lambda_A$'s span a subspace in $\hh$, which we denote $\hh^D$. In other words,
\begin{equation}\label{psi}
\hh^D=\mathrm{span}\big\{\psi \in \hh\;|\;\qq(d_A)\psi=\lambda_A\psi\,, \; \lambda_A>0\,,\;A=1,\ldots,n\;\big\} \,.
\end{equation}
The presence of the Laplace operator implies that $\dim \hh^D <\infty$. The Weyl law (\ref{WR}) approximates the dimension of $\hh^D$ for $R/\hbar\gg 1$. By abuse of language, we will refer to $\psi$ and $\hh^D$ from (\ref{psi}) as the eigenstate and eigenspace of the domain $D$. If $D_1$ and $D_2$ are two quantizable basic domains that are compatible with each other,  then $D_1\cup D_1$ is quantizable and we define
\begin{equation}\label{DD}
 \hh^{D_1\cup D_2}:=\hh^{D_1}+\hh^{D_2}\,.
\end{equation}
Since all the Hermitian operators $\qq(d_A)$ associated with the compatible domains $D_1$ and $D_2$ are pairwise commuting, we can choose an orthonormal basis in $\hh^{D_1\cup D_2}$ by combining  orthonormal bases in $\hh^{D_1}$ and $\hh^{D_2}$.   Using this diagonalizing basis shows that 
$$
\hh^{D_1\cap D_2}=\hh^{D_1}\cap\hh^{D_2} 
$$
and 
$$
\hh^{D_1}\cap\hh^{D_2}=0\quad \Longleftrightarrow\quad \hh^{D_1}\perp\hh^{D_2}\,.
$$
In the last case, we say that the domains $D_1$ and $D_2$ are $\hbar$-disconnected. Conversely, the domains $D_1$  and $D_2$ are said to be $\hbar$-connected if they share a common eigenstate. Let us emphasize that the topological connectedness of the domains ($D_1\cap D_2\neq \varnothing$) does not necessarily imply their $\hbar$-connectedness and two compatible domains may well be $\hbar$-connected for one value of $\hbar$ and $\hbar$-disconnected for another.  For $\hbar$-disconnected domains $D_1$ and $D_2$, the sum in (\ref{DD}) becomes direct.

The space $\hh^D$ is identified with the state space of a quantum-mechanical system whose classical phase space is  $G\times D$. Since the factor $G$ represents the space of momenta, we can speak of a momentum representation for the space of states.  Let $P_D$ denote the Hermitian projector onto the  subspace $\hh^D\subset \hh$. By definition,
\begin{equation}
    P_D^\dagger=P_D\,,\qquad P_D^2=P_D\,.
\end{equation}
Given a classical observable $a\in \fa$, we set
\begin{equation}
    \qq^D(a)=P_D\qq(a)P_D\,.
\end{equation}
Since $\mathrm{Im} \big(\qq^D(a)\big)\subset \hh^D$, we can think of $\qq^D(a)$ as an operator on the space $\hh^D$. 
This defines a linear map $\qq^D:\fa\rightarrow \mathscr{L}(\hh^D)$ from the space of classical observables to the algebra of linear operators on $\hh^D$, which we call the {\it quantization map}. If $\dim \hh^D=N$, then we will write $\hh^D_N$ to indicate the dimension of $\hh^D$.   Following the pattern of the cylinder from Sec. \ref{RecDom}, we endow the finite-dimensional space  $\mathscr{L}(\hh^D_N)$  with the normalized Frobenius norm:  
\begin{equation}
    ||A||_N^2=\frac1N{\mathrm{Tr}}(A^\dagger A)\,,\qquad \forall A\in \mathscr{L}(\hh^D_N)\,.
\end{equation}

\begin{definition}[Bulk and Boundary States]
Let \(\mathcal{H}^D\) be the state space associated with a quantizable domain \(D\), and let
$$
\hh^D_{(t)}\supset \hh_{(t')}^D\,,\qquad t\leq t',\qquad \mathcal{H}^D_{(t)} = \left\{ \psi \in \hh^D \ \middle|\ \qq(a)\psi \in \mathcal{H}^D \quad \forall\, a \in \fa^{(t)} \right\}
$$
be the descending filtration induced by the total filtration \(\fa^{(t)}\) in the space of classical observables.
The elements of \(\mathcal{H}^D_{(t)}\) are called {bulk states of filtration degree \(t\)}, by analogy with the interior points of \(D\).  
Correspondingly, the elements of the orthogonal complement $(\mathcal{H}^D_{(t)})^\perp$ are referred to as {boundary states}.
\end{definition}

\begin{definition}[Thick Domain]
A quantizable domain \(D\) is called {thick} if, for any finite \(t\),
\[
\lim_{\hbar \to 0} \frac{\dim \mathcal{H}^D_{(t)}}{\dim \mathcal{H}^D} = 1\qquad \Longleftrightarrow\qquad \lim_{\hbar \to 0} \frac{\dim (\mathcal{H}^D_{(t)})^\perp}{\dim \mathcal{H}^D} = 0\,.
\]
In other words, in the classical limit \(\hbar \to 0\), the proportion of boundary states with any given filtration degree becomes negligible compared to the bulk states of that filtration degree.
\end{definition}

\begin{figure}
    \centering
\begin{tikzpicture}[scale=1.2]
\draw[] (0, 0) circle (2);
\draw[fill=blue!20, opacity=1] (25:2) arc(25:155:2) -- cycle;
\draw[fill=green!20, opacity=1] (40:1.72) arc(40:140:1.72) -- cycle;
\draw[->] (0,-2.5) -- (0,2.5);
\draw[->] (-2.5,0) -- (2.5, 0);
\draw (2.5, 0) node [below] {${}_{x_1}$};
\draw (0, 2.4) node [left] {${}_{x_3}$};

\draw [-] (7,-1)--(10,2);
\draw [-] (7,-1)--(4,2);

\draw[-, fill=blue!20, opacity=1]  (9.5,1.5)--(7.5,-0.5)--(7.5,1.5)--cycle;

\draw[-, fill=green!20, opacity=1]  (9.13,1.13)--(7.9,-0.1)--(7.9,1.13)--cycle;

\draw[->] (7,-2.5) -- (7,2.5);
\draw[->] (4,-1) -- (10, -1);

\foreach \y in {0,...,14}
\foreach \x in {0,...,\y}
  \draw[fill=black] (\x/2.5 - \y/5+7,\y/5-1) circle (0.5pt);

\draw (7, 2.4) node [left] {${}_{j}$};
\draw (9.9, -1) node [below] {${}_{m}$};
\end{tikzpicture}
 \caption{ Left panel: a spherical cap in $\mathfrak{su}(2)^\ast\simeq\mathbb{R}^3$. Right panel: the weight diagram of the right regular  representation of $SU(2)$. The green and blue regions indicate the bulk and the boundary states, respectively.}
    \label{Pic2}
\end{figure}
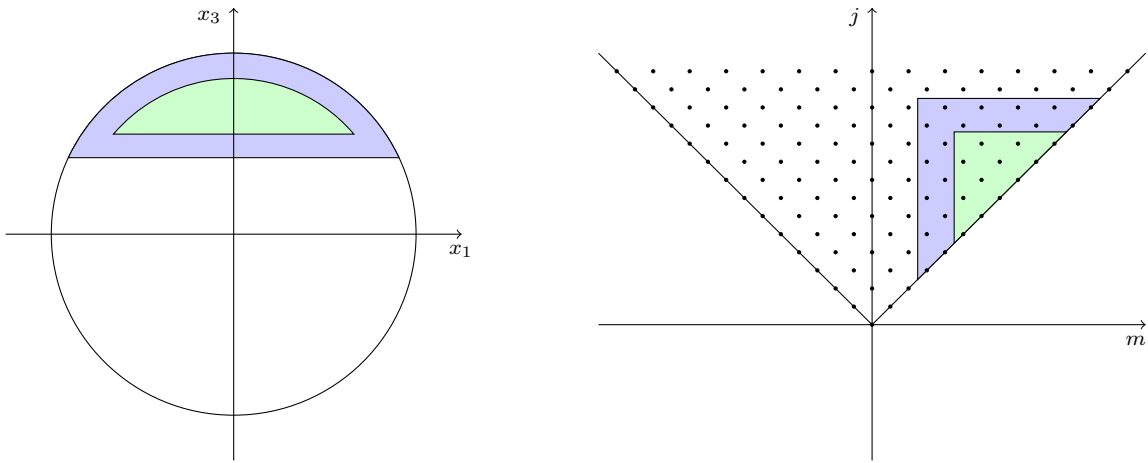

\begin{example}Consider the Lie algebra $\mathfrak{g}=\mathfrak{su}(2)$ and take as the bounded domain  $D$ the spherical  cap defined by the polynomials
\begin{equation}
   d_1=R^2-x^2_1-x^2_2-x^2_3\,,\qquad d_2=x_3-h\,,
\end{equation}
where $h<R$. The pre-Hilbert space $\hh=\mathcal{F}(SU(2))$ carries the right regular representation of $SU(2)$ generated by $\qq(x_a)$. This action preserves the Casimir filtration on $\mathcal{F}(SU(2))$, so that the associated graded space defines the decomposition $\hh$ into the direct sum of invariant subspaces:
\begin{equation}
    \hh=\bigoplus_{j\in \frac12\mathbb{N}_0}\hh_j\,,\qquad    \hh_j \ni \psi \quad \Longleftrightarrow\quad \mathbb{L}\psi=j(j+1)\psi\,.\end{equation}
Each subspace $\hh_j$ in turn decomposes into $2j+1$ invariant subspaces of the same dimension. These carry the irreducible representations of spin $j$, so that $\dim\hh_j=(2j+1)^2$. The eigenstates of the domain  $D(d_1, d_2)$ can be conveniently depicted on the weight diagram of the right regular representation of $SU(2)$ (see  Fig. \ref{Pic2}).  It is easy to compute that the total number of weights associated with the cap region with $R^2=\hbar^2s(s+1)$ and $h=\hbar m$ is 
\begin{equation}
    N=\sum_{j=m+1/2}^s (2j+1)(2j-2m)= \frac{2}{3}{(s-m)(2s-2m+1)\left(2s + m +2\right)}\,.
\end{equation}
In the classical limit $\hbar\rightarrow 0$ this yields
\begin{equation}
    N\sim \frac4{3\hbar^3}\big[{(R-h)^2(2R+h)}+\mathcal{O}(\hbar)\big]\,.
\end{equation}
For $h=-R$, we reproduce the Weyl law (\ref{WR}). Now, for each $t$, consider a smaller cap region $D_t\subset D$ defined by the polynomials: 
\begin{equation}
    d^t_1=R_t^2-x_1^2-x_2^2-x_3^2\,,\qquad d^t_2=x_3-h_t\,,
\end{equation}
$$R_t=R-\hbar t\,,\qquad h_t=h+\hbar t\,.$$ Using the Cartan--Weyl basis for $\mathfrak{su}(2)$ and the total filtration for $\fg$,  one can easily see that the state space 
$\hh^{D_t}$ is spanned by the bulk states of the initial spherical cap $D(d_1,d_2)$. Since in the classical limit $R_t\rightarrow R$ and $h_t\rightarrow h$, we immediately conclude that 
\begin{equation}
    \lim_{\hbar\rightarrow 0}\frac{\dim \hh^{D_t}}{\dim \hh^D}=\lim_{\hbar\rightarrow 0}\frac{N_t}{N}=1\,.
\end{equation}
Hence, the spherical cap defines a thick domain. 

\end{example}

It is not hard to extend these examples to more general thick domains determined by higher-order Casimir elements and vectors of the Cartan subalgebra, but we will not dwell on that here. 
 
\begin{theorem} For any quantizable domain $D$ the following properties hold: 
    \begin{enumerate}
        \item[i)] normalization: 
                      $$\qq^D(1)=\I\,,$$
        \item[ii)] reality: 
        $$\big(\qq^D(a)\big)^\dagger=\qq^D(\bar a)\,.$$
         \end{enumerate}   
         If the domain $D$ is thick,  then one additionally has
         \begin{enumerate}
         \item [iii)] von Neumann's condition: 
         $$  \lim_{\hbar\to 0}\big\|\qq^D(a)\qq^D(b)-\qq^D(ab)\big\|_N=0\,.$$
        \item [iv)] Dirac's condition: 
        $$  \lim_{\hbar\to 0}\big\|\big((i\hbar)^{-1}[\qq^D(a),\qq^D(b)]-\qq^D(\{a,b\})\big)\psi\big\|=0\qquad \forall a,b\in \fa_{(t)},\quad \forall \psi\in \hh^D_{(2t)}\,.$$
        
    \end{enumerate}
      \end{theorem}

\begin{proof}
Let us denote
$$
A:=\qq^D(a)\qq^D(b)-\qq^D(ab)\,,
$$
for some $a,b\in \mathfrak{A}_{(t)}$. We need to evaluate the classical limit of $\|A\|_{N}$. Let $\psi_1,\psi_2,\ldots, \psi_n$ be a basis  in the subspace $\hh^D_{(2t)}$ of bulk states of filtration degree $2t$. We complete this basis to a full basis $\psi_1,\ldots, \psi_N$ in $\hh^D$, where the additional $N-n$ basis vectors represent the boundary states. Then 
\begin{equation}\label{Apsi}
\|A\|^2_N=\frac1N\sum_{k=1}^N\|A\psi_k\|^2=\frac1N\sum_{k=1}^n\|A\psi_k\|^2+\frac1N\sum_{k=n+1}^N\|A\psi_k\|^2\,.
\end{equation}
Since $a, b, ab\in \mathfrak{A}_{(2t)}$, we have 
\begin{equation}
    A\psi_k=\big(\qq(a)\qq(b)-\qq(ab)\big)\psi_k\,,\qquad k=1,\ldots,n\,.
\end{equation}
By Proposition \ref{P3}, the first sum on the right of (\ref{Apsi}) is evaluated as follows:
\begin{equation}
\begin{array}{c}
   \displaystyle \frac1N\sum_{k=1}^n\|A\psi_k\|^2= \frac1{N}\sum_{k=1}^n  \Big\| \sum_{m=1}^{[2t]}\hbar^{m}\qq\big(D_m(a,b)\big)\psi_k\Big\|^2\\[5mm]
 \displaystyle  \leq \frac{n}{N}\Big(\sum_{m=1}^{[2t]}\hbar^{m}\big\| \qq\big(D_m(a,b)\big)\big\|\Big)^2   \leq  \frac{n}{N}\Big(\sum_{m=1}^{[2t]}\hbar^{m}\big|D_m(a,b)\big|_R\Big)^2 \quad \longrightarrow 0
 \end{array}
    \end{equation}
as $\hbar\to 0$. For the second sum, we have 
\begin{equation}
\begin{array}{c}
\displaystyle \frac1N\sum_{k=n+1}^N\|A\psi_k\|^2\leq \frac{N-n}{N}\|A\|^2\\[5mm]
\displaystyle \leq \frac{N-n}N\big(\|\qq(a)\|\cdot \|\qq(b)\|+\|\qq(ab)\|\big)^2
\leq 4|a|^2_R|b|^2_R\frac{N-n}{N}\,.
\end{array}
\end{equation}
By the thickness property,  $(N-n)/N\rightarrow 0$ in the classical limit. 

To prove the Dirac condition, we introduce the operator 
\begin{equation}
B=(i\hbar)^{-1}[\qq^D(a),\qq^D(b)]-\qq^D(\{a,b\})\,.
\end{equation}
For bulk states, we have  
\begin{equation}
 \|B\psi \|=\big\|\big((i\hbar)^{-1}[\qq(a),\qq(b)]-\qq(\{a,b\})\big)\psi\big\|\,,\qquad\forall a,b\in \fa_{(t)}\,,\quad \forall \psi\in \hh^D_{(2t)}\,.
    \end{equation}
Applying then  Proposition \ref{P3} yields 
\begin{equation}
\begin{array}{c}
  \displaystyle  \|B\psi\|=\sum_{m=2}^{[2t]}\hbar^{m-1}\big\| \qq\big(D_m(a,b)-D_m(b,a)\big)\psi\big\|  \\[3mm] 
    \displaystyle  \leq \sum_{m=2}^{[2t]}\hbar^{m-1}\big\| \qq\big(D_m(a,b)-D_m(b,a)\big)\big\|   \leq  \sum_{m=2}^{[2t]}\hbar^{m-1}\big|D_m(a,b)-D_m(b,a)\big|_R \quad \longrightarrow 0    \end{array}
    \end{equation}
as $\hbar\to 0$, so the Dirac condition follows. 

As is seen, the thickness property is essential here. For thick quantizable domains, the quantum operators 
$\qq^D(a)$ correctly reproduce the classical algebra in the semiclassical limit: products and commutators converge to their classical counterparts. This ensures that the quantization is consistent and physically meaningful, with boundary effects vanishing in the classical regime.
\end{proof}

At this point, one may wonder about the analogue of the separability condition ($iii$) from Proposition \ref{P1}. This condition does not generally appear to hold for arbitrary compact groups and physical observables. However, below we present an important case where it does apply. 

\begin{example}
    Consider an open  ball $B^d_R\subset\gs$ defined by the quadratic Casimir element:
    \begin{equation}
        R^2-x_1^2-\cdots -x_d^2>0\,.
    \end{equation}
   Then for any $f\in \ff_t(G)$ we can write
    \begin{equation}\label{sum}
        \|\qq^{D}(f)\|^2_N=\frac1N\sum_{|\pi|<l}\dim \pi\sum_{i,j=1}^{\dim \pi}\int_Gdg\,\bar \pi_{ij}(g)\bar f(g)P_D \big(f(g)\pi_{ij}(g)\big)\,.
    \end{equation}
    Here the outer sum runs over all (equivalence classes of) representations $\pi\in \hat{G}$ whose Casimir length is less than $l=R/\hbar$. Consequently, the dimension of the state space $\hh_N^D=\ff_l(G)$ is  
    \begin{equation}\label{NN}
        N=\sum_{|\pi|<l}(\dim \pi)^2 \,.
        \end{equation}

Taking into account the Casimir length of $f$, we can split the outer sum in (\ref{sum}) into two parts: one corresponding to bulk states and the other to boundary states of filtration degree $t$:
\begin{equation}\label{1N}
\frac1N\sum_{|\pi|<l}(\ldots)=\frac1N\sum_{|\pi|<l-t}(\ldots)+\frac{1}{N}\sum_{ l-t\leq |\pi|<l}(\ldots)\,.
\end{equation}
The classical limit of the first sum is computed as
\begin{equation}
\begin{array}{c}
 \displaystyle  \lim_{\hbar\to 0} \frac1N\sum_{|\pi|<l-t}\dim \pi\sum_{i,j=1}^{\dim \pi}\int_Gdg\,\bar \pi_{ij}(g)\bar f(g)f(g)\pi_{ij}(g)\\[5mm]
 \displaystyle  =\lim_{\hbar\to 0} \frac1N\sum_{|\pi|<l-t}(\dim \pi)^2\int_Gdg|f(g)|^2=\lim_{\hbar\to 0} \frac{\sum_{|\pi|<R/\hbar-t}(\dim \pi)^2}{\sum_{|\pi|<R/\hbar}(\dim \pi)^2 
 }\int_Gdg|f(g)|^2\\[5mm]\displaystyle =\int_Gdg|f(g)|^2= \|f\|^2_D\,,
 \end{array}
\end{equation}
since the ratio tends to $1$ by the Weyl law (\ref{WR}).
Here we also used the fact that $\pi$ is  a unitary representation of $G$, so that $\bar \pi_{ij}(g)=\pi_{ji}(g^{-1})$ and
\begin{equation}
    \sum_{i,j=1}^{\dim \pi} \bar\pi_{ij}(g)\pi_{ij}(g)=\dim \pi\,.
\end{equation}
As to the second sum in (\ref{1N}), it vanishes in the classical limit due to the thickness of the ball domain. Indeed,
\begin{equation}
\begin{array}{c}
    \displaystyle  \lim_{\hbar\rightarrow0}\frac1N\sum_{l-t\leq|\pi|<l}\dim \pi\sum_{i,j=1}^{\dim \pi}\int_Gdg\,\bar \pi_{ij}(g)\bar f(g)P_D f(g)\pi_{ij}(g)\\[5mm]
  \displaystyle  =  \lim_{\hbar\to 0}\frac1N\sum_{l-t\leq|\pi|<l}\dim \pi\sum_{i,j=1}^{\dim \pi}\|\qq^D(f)\pi_{ij}\|^2\leq\|\qq(f)\|^2 \lim_{\hbar\rightarrow 0}\frac1N   \sum_{R/\hbar-t\leq|\pi|<R/\hbar}(\dim \pi)^2 \\[5mm]
 \displaystyle  \leq|f|_\infty^2\lim_{\hbar\to 0}\frac{(R-\hbar t)^d-R^d}{R^d}=0\,.
  \end{array}
      \end{equation}
In the last step, we used (\ref{NN}) together with the Weyl law (\ref{WR}). Thus, 
\begin{equation}\label{=}
    \lim_{\hbar\to 0}\|\qq^D(f)\|^2_N=\|f\|^2_D
\end{equation}
for $D=B_R^d$ and $f\in \ff(G)$. 

From a physical viewpoint, one can regard $f$ as the Hamiltonian of a free particle whose  phase space is $B_R^d\times G$. In this interpretation, equality (\ref{=}) states that the classical limit of the  root mean square value of the quantum energy spectrum coincides with that of the classical Hamiltonian $f(g)$. This shows that the quantization scheme respects the correspondence principle: the statistical properties of the spectrum approach those of the underlying classical system once quantum fluctuations are suppressed.

\end{example}

Using the concept of bulk and boundary states, it is not hard to generalize Proposition \ref{P2} from one-dimensional intervals to  arbitrary balls in higher dimensions.

\begin{proposition}
    Let $D=B^d_R\subset\gs$ be an open ball and let $\hh$ be the space of representative functions on $G$, or its $L^2$-completion. Then for any 
    $a,b\in \fa$ and $\psi\in \hh$, the quantization map asymptotically reproduces classical multiplication and Poisson brackets in the sense that the following limits hold:
    \begin{equation}\label{cllim}
    \begin{array}{l}
     \displaystyle   \lim_{\hbar\to 0} \big\|  \big(\qq^D(a)\qq^D(b)-\qq^D(ab)\big) \psi  \big\|=0\,,\\[5mm]
   \displaystyle     \lim_{\hbar\to 0} \big\|  \big((i\hbar)^{-1}[\qq^D(a),\qq^D(b)]-\qq^D(\{a,b\})\big) \psi  \big\|=0\,.
       \end{array}
       \end{equation}
\end{proposition}

The proof exploits the Casimir filtration in $\fg$, which allows one to represent each function of $\hh$ as a generalized Fourier series 
\begin{equation}
    \psi(g)=\sum_{l}\sum_{|\pi|=l}\psi_{\pi}^{ij}\sqrt{\dim\pi}\,\pi_{ij}(g)\,,
\end{equation}
where the inner sum runs over all (equivalence classes of) representations $[\pi]\in \hat{G}$ of Casimir length $l$.  Square summability, 
\begin{equation}
    \|\psi\|^2=\sum_l\sum_{|\pi|=l}\sum_{i,j}|\psi_\pi^{ij}|^2<\infty\,,
\end{equation}
implies that $|\psi_{\pi}^{ij}|\to 0$ as $|\pi|\to\infty$. This decay condition suppresses the contribution of boundary states to the classical limits (\ref{cllim}).

\begin{theorem}Let $D(d_1,\ldots,d_n)$ be a bounded and quantizable semialgebraic  domain and let $f\in \mathrm{Pol}(\gs)$ be a polynomial strictly positive on the topological closure of 
$D$, i.e., $f|_{\bar D}>0$.  Then, for each $\psi \in \hh^D$ and each $\varepsilon>0$, there exists a constant $c>0$ such  that 
   $$
   \big(\psi, \qq^D(f)\psi\big) >-\varepsilon 
   $$
   for all $\hbar<c$. 
\end{theorem}

\begin{proof}
By Theorem \ref{T1}, any polynomial $f(x)$ positive on $\bar D$ belongs to the semiring $S(D)$, meaning that 
\begin{equation}
    f=\sum_{e_k\in \{0,1\}} f^2_{e_1e_2\cdots e_d}d_1^{e_1}d_2^{e_2}\cdots d_n^{e_n}
\end{equation}
for some polynomials $f_{e_1e_2\cdots e_k}(x)$. Applying the prequantization map, we can bring the operator $\qq(f)$ into the form
$$
\qq(f)=\sum_{e_k\in \{0,1\}} \qq(f_{e_1\cdots e_n})\big(\qq(d_{1})\big)^{e_1}\cdots \big(\qq(d_n)\big)^{e_n}\qq(f_{e_1\cdots e_n})+\hbar \qq(r_\hbar)\,,
$$
where $r_\hbar$ is some polynomial in the $x$'s and $\hbar$. Then, for any normalized state $\psi\in \hh^D$,
\begin{equation}\label{sum}
\begin{array}{c}
  \displaystyle   \big(\psi, \qq^D(f)\psi\big)=\sum_{e_k\in \{0,1\}} \big(\qq(f_{e_1\cdots e_n})\psi, \big(\qq(d_1)\big)^{e_1}\cdots \big(\big(\qq(d_n)\big)^{e_n}\qq(f_{e_1\cdots e_n})\psi\big) \\[5mm]+\hbar \big(\psi, \qq(r_\hbar)\psi\big)\,.
    \end{array}
\end{equation}
By the Schwarz inequality and Rel. (\ref{QR}), 
\begin{equation}
    \big|\big(\psi,\qq(r_\hbar)\psi\big)\big|^2\leq \|\qq(r_\hbar)\|^2\leq \|\qq(r_\hbar)\|^2\leq |r_\hbar|^2_R\,,
\end{equation}
where $R$ is the radius of the domain $D$. Since the operators $\big(\qq(d_1)\big)^{e_1}\cdots \big(\big(\qq(d_n)\big)^{e_n}$ are positive semidefinite on $\hh^D$, the sum over the $e_k$'s in (\ref{sum}) is non-negative, while the last term tends to zero as $\hbar\rightarrow 0$. Therefore, we can take $c =\varepsilon/(b+\varepsilon) \leq 1$, where
$$
b=\max_{0\leq \hbar\leq 1}|r_\hbar|_R\,.
$$
\end{proof}

Heuristically, the theorem shows that a quantum particle cannot be detected far beyond the domain $D$, where the observable $f$ may take large negative values. In the classical limit, the particle becomes increasingly confined within
$D$. Of course, since the position operators do not commute, the particle's location in space is intrinsically uncertain. Consistently, inequality (\ref{qx}) indicates that the possible measurement outcomes of the particle’s coordinates 
$x^a$ do not exceed the radius of the domain $D$. 

\section{Discussion}

Let us briefly summarize the main results and indicate possible directions for future work. In this paper, we developed a quantization scheme for bounded symplectic domains associated with compact Lie groups. Specifically, we studied domains of the form $D\times G\hookrightarrow T^\ast G$, where $G$ is a compact Lie group and $D$ is a bounded region in the dual space $\mathfrak{g}^\ast$ of the corresponding Lie algebra. The finiteness of the symplectic volume implies that the resulting quantum theory is described by a finite-dimensional Hilbert space, with observables represented by finite Hermitian matrices.
Particular attention was devoted to a detailed understanding of the classical limit.

It is worth emphasizing that most of the existing literature on quantization focuses on smooth symplectic manifolds, whether compact or noncompact. In contrast, in our setting, the presence of the boundary $\partial \bar D \times G$ gives rise to specific boundary effects that cannot be ignored. These effects require a careful reconsideration -- and indeed a modification -- of the standard von Neumann and Dirac conditions that usually guide the quantization procedure.
As a consequence, the quantization scheme developed here does not coincide with any single known approach. Rather, it combines elements of several: deformation quantization at the prequantization stage, projection onto the state space as in Berezin--Toeplitz quantization, and the construction of the operator norm in the spirit of strict quantization. Moreover, it reveals noteworthy links with semialgebraic geometry. This hybrid character highlights both the novelty of the framework and the richness of the underlying mathematical structures. 

A natural direction for further research is to extend the present scheme by analyzing the separability condition from Proposition \ref{P1} in the context of more general quantizable domains and observables. In particular, a systematic analysis of the instances and forms in which separability holds would further clarify the nature of the classical limit. Equally intriguing would be to explore possible connections with {\it noncommutative} semialgebraic geometry, where related structural phenomena may appear in a broader algebraic framework \cite{helton2004positivstellensatz, schmudgen2008noncommutative}.

Finally, the proposed quantization scheme appears promising from the perspective of physical applications, as it provides a consistent framework for quantizing systems with compact momentum {\it and} configuration spaces. In this setting, the phenomenon of supertunneling, identified here, appears to be new and merits further investigation. Equally significant is the emergence of a maximal density of the fermionic gas, pointing to novel physical effects that may arise in systems with compact momentum space. In particular, the concept of a maximal fermion density could offer a natural resolution of the singularity problem in the collapse of massive stars and other compact astrophysical objects; in this scenario, the singularity is avoided not through speculative quantum gravity effects (for which no complete theory yet exists), but rather through the model-independent Pauli exclusion principle built into quantum mechanics itself.  The limiting volume is determined by the product of the minimal spatial cell, $(2\pi \hbar)^3/\mathrm{Vol}(G)$, and the number $N$ of fermions involved in the collapse. We intend to explore these questions in future work.

\section*{Acknowledgements}This work was supported by the research project FSWM-2025-0007 of the Ministry of Science and Higher Education of the Russian Federation.
\appendix 

\section{Proofs of Propositions \ref{P1} and \ref{P2}}\label{A}

It is convenient to introduce the operator norm on $\mathscr{L}(\hn)$. Let $\|\psi\|$ denote the norm of a state  $\psi\in \hn$, then for every $A\in\mathscr{L}(\hn)$ we set
\begin{equation}
    \|A\|=\max_{\|\psi\| =1}\| A\psi\| \,.
\end{equation}
For example, 
\begin{equation}\label{zf}
    \|\widehat {z^n}\|=1\quad \mbox{for}\quad |n|<N
 \qquad \mbox{and}\qquad \|f(\hat x)\|\leq |f|_\infty\,,
\end{equation}
where $f(x)$ is any polynomial and 
\begin{equation}
    |f|_\infty=\max_{x\in [0,L]} |f(x)|\,.
\end{equation}
Both relations (\ref{zf}) follow directly from the definition of the operators $\widehat {z^n}$ and $f(\hat x)$  in the orthonormal basis (\ref{base}): 
\begin{equation}
    \widehat {z^n}\psi_k=\chi_N(k+n)\psi_{k+n}\,, \qquad   f(\hat x)\psi_k =f(x_k)\psi_k \,.
    \end{equation} 
Thus, the operators $\widehat {z^n}$ and $f(\hat x)$ are  uniformly bounded on the sequences of spaces $\hn$. Since the eigenvalues $\{x_1,x_2,\ldots, x_N\}\subset (0,L]$ of the operator $\hat x$ constitute a dense subset as $N\rightarrow\infty$, we conclude that
\begin{equation}\label{fN}
    \lim_{N\rightarrow \infty}\|f(\hat x)\|=|f|_\infty\,.
\end{equation}

The operator and normalized Frobenius norms are related by the inequality 
\begin{equation}
    \|A\|_N\leq \|A\|\,,\qquad \forall{A\in \mathscr{L}(\hn)}\,.
\end{equation}
As mentioned in footnote \ref{f1}, the normalized Frobenius norm is not submultiplicative, in contrast to the operator norm. However, a kind of submultiplicativity can be restored if we consider both norms together. Specifically, the following inequalities hold \cite{vonNeumann1942, CIT-006}:
\begin{equation}\label{sm}
\begin{array}{ll}
   \|AB\|_N\leq \|A\|\cdot\|B\|_N\,, &\qquad  \|AB\|_N\leq \|B\| \cdot\|A\|_N\,, \\[5mm] 
  \|AB\|\leq \|A\|\cdot\|B\|\,, &\qquad \|AB\|_N\leq \|B\| \cdot\|A\|\,.\end{array}
  \end{equation}
With these preliminaries, we are ready to prove Propositions \ref{P1} and \ref{P2}.

\paragraph{Separability condition.} Consider first the case where the classical observable is given by a polynomial $f(x)$. Using the orthonormal basis $\big\{\psi_k(p)=e^{\frac{2\pi i k p}{M}}\big\}_{k=1}^N$ in $\hn$, we find
\begin{equation}
    \|\hat f\|^2_N=\frac{1}{N}\sum_{k=1}^N(f(\hat x)\psi_k,f(\hat x)\psi_k) =\frac{1}{N}\sum_{k=1}^N(f(x_k)\psi_k,f(x_k)\psi_k)  =\frac{1}{N}\sum_{k=1}^N |f(x_k)|^2\,,
\end{equation}
whence 
\begin{equation}\label{fN}
\displaystyle    \lim_{N\rightarrow \infty}  \|\hat f\|^2_N=\frac1L\lim_{N\to\infty} \frac{L}{N}\sum_{k=1}^N |f(x_k)|^2 =\frac1L\int_0^L|f(x)|^2dx=\|f\|^2_D\,.
\end{equation}
A general classical observable of $\fa$ is represented by a finite series 
             \begin{equation}
             a(x,p)=\sum_{n\in\mathbb{Z}} a_n (x) z^n\,,
             \end{equation}
where only finitely many polynomials $a_n(x)$ are different from zero. Its norm squared is
             \begin{equation}
                 \|a\|^2_D=\frac1{LM}\int |a(x,p)|^2dxdp=\frac1{L}\sum_{n\in \mathbb{Z}}\int |{a}_n(x)|^2dx\,.
             \end{equation}
             It follows from the commutation relation $[\hat x,\hat z ]=(2\pi\hbar/M)\hat{z}$ that for any polynomial $f(x)$,
             \begin{equation}\label{fz}
                 f(\hat x)\hat{z}^n= \hat{z}^nf(\hat x+n\Delta x)\,,  \qquad \Delta x:=\frac{2\pi\hbar}{M}=\frac LN\,.  
             \end{equation}
    This identity allows us to express any quantum observable in $zx$-order form as 
        \begin{equation}\label{zx}
                \hat{a}=\sum_{n\in \mathbb{Z}}\hat{z}^n a^\hbar_n(\hat{x}\big)\,, 
                \end{equation}
                where  
                \begin{equation}\label{ah}
                    {a}^\hbar_n(x):=a_n(x)+\Delta a_n(x)\,,\qquad \Delta a_n(x):=\frac12\big(a_n(x+n\Delta x)-a_n(x)\big)\,.
                \end{equation}
Clearly,  $\Delta x\rightarrow 0$ and $a^\hbar_n(x)\rightarrow a_n(x)$ as $\hbar\rightarrow 0$. 
                
On the other hand, 
             \begin{equation}
             \begin{array}{c}
              \displaystyle   \|\hat{a}\|^2_N=\frac1{N}\sum_{k=1}^N (\hat a\psi_k,  \hat a\psi_k)=  
                 \frac1{N}\sum_{k=1}^N\sum_{n,m\in\mathbb{Z}} \big(\hat{z}^na^\hbar_n(x_k)\psi_k,  \hat{z}^ma^\hbar_m(x_k)\psi_k\big) \\[7mm]               
 \displaystyle  =\frac1{N}\sum_{k=1}^N\sum_{n,m\in\mathbb{Z}}\overline {a^\hbar_n(x_k)}a^\hbar_m(x_k)\big(\hat{z}^n\psi_k,  \hat{z}^m\psi_k\big) \\[7mm]
 \displaystyle  =\frac1{N}\sum_{k=1}^N\sum_{n,m\in\mathbb{Z}}\overline{ a^\hbar_n(x_k)}a^\hbar_m(x_k)\chi_N(k+n)\chi_N(k+m)\big(\psi_{k+n},  \psi_{k+m}\big)\\[7mm]
 \displaystyle =\frac1{N}\sum_{n\in\mathbb{Z}}\sum_{k=1}^{N} \chi_N(k+n)|a^\hbar_n(x_k)|^2  \,,                   
 \end{array}
                   \end{equation}
so that 
     \begin{equation}
  \lim_{N\to \infty}\|\hat{a}\|^2_N=\frac1{L}\sum_{n\in\mathbb{Z}}\lim_{N\rightarrow \infty}\sum_{k=1}^{N}\frac{L}{N}\chi_N(k+n)|a^\hbar_n(x_k)|^2 
    =\frac1{L}\sum_{n\in\mathbb{Z}}\int_0^L|a_n(x)|^2 dx =\|a\|_D^2\,.  \end{equation}       

\paragraph{Von Neumann's condition.} 
Let us first consider the special case where $a=z^n$ and $b=z^m$. Then
  \begin{equation}
      \|\hat z^n \hat{z}^m-\widehat{z}^{n+m}\|^2_N=\frac{1}{N}\sum_{k=1}^N (\varphi_k,\varphi_k)\,,
  \end{equation}
  where 
  \begin{equation}
      \varphi_k=(\hat z^n \hat{z}^m-\widehat{z}^{n+m})\psi_k=\chi_N(k+n+m)(\chi_N(k+m)-1)\psi_{k+n+m}\,.
      \end{equation}
 Hence,
       \begin{equation}\label{zzz}
      \|\hat z^n \hat{z}^m-\widehat{z}^{n+m}\|^2_N\leq \frac{|m|}{N}\,,\qquad \mbox{so that }\qquad \lim_{N\rightarrow\infty}\|\hat z^n \hat{z}^m-\widehat{z}^{n+m}\|_N=0\,. 
  \end{equation}
  By applying the triangle inequality for the norm, we observe that the operators $\hat z^n$ and $\hat z^m$ asymptotically commute:
  \begin{equation}\label{zzN}
      \lim_{N\to\infty}\|\hat z^n\hat z^m-\hat z^m\hat z^n\|_N=0\,.
  \end{equation}
  
  Next, consider the general case where the classical observables are 
 \begin{equation}\label{ab}
  a=\sum_{n\in\mathbb{Z}}z^na_n(x)\,,\qquad b=\sum_{m\in\mathbb{Z}}z^mb_m(x)\,.
\end{equation}
  To establish von Neumann's property, it suffices  to examine the individual terms of these polynomials, taking advantage of the norm's subadditivity.  
  Using  (\ref{zx}) and (\ref{ah}), we can write the corresponding quantum observables as 
  \begin{equation}\label{za}
      \widehat{z^na_n(x)}=\hat z^n a_n^\hbar(\hat x)\,,\qquad  \widehat{z^mb_m(x)}=\hat z^m b_m^\hbar(\hat x)  \,,
      \end{equation}
      which leads to 
      \begin{equation}\label{dif}
      \begin{array}{c}
          (\widehat{z^na_n} ) (\widehat{z^mb_m})- \widehat{(z^{n+m}a_nb_m)}  \\[5mm]
     =\hat z^n\hat z^ma_n^\hbar(\hat x+m\Delta x)b_m^\hbar(\hat x)-\hat z^{n+m}a^\hbar_n(\hat x+m\Delta x) b_m^\hbar(\hat x+n\Delta x) \\[5mm]
      =\big[\hat z^n\hat z^m -\hat z^{n+m}\big]a_n^\hbar(\hat x+m\Delta x)b_m^\hbar(\hat x)+\hat z^{n+m}a^\hbar_n(\hat x+m\Delta x) \big[b_m^\hbar(\hat x)-b_m^\hbar(\hat x+n\Delta x) \big ]   \,.  
      \end{array}
      \end{equation}      The operators in square brackets can be regarded as ``asymptotically small''. Indeed, using subadditivity and submultiplicativity  (\ref{sm}) of the norm,  we can write 
      \begin{equation}
      \begin{array}{c}
            \|(\widehat{z^na_n} ) (\widehat{z^mb_m})- \widehat{z^{n+m}a_nb_m}\|_N   \leq \|a_n^\hbar(\hat x+m\Delta x)b_m^\hbar(\hat x)\| \cdot \|\hat z^n\hat z^m -\hat z^{n+m} \|_N\\[5mm]
            +\|\hat z^{n+m}\|\cdot\|a^\hbar_n(\hat x+m\Delta x)\|\cdot \|b_m^\hbar(\hat x)-b_m^\hbar(\hat x+n\Delta x)  \|\,.
            \end{array}
            \end{equation}
Applying  Eqs. (\ref{zf}) and (\ref{zzz}) to the terms on the right yields 
    $$ \begin{array}{l}
       \displaystyle  \lim_{N\to \infty}\|a_n^\hbar(\hat x+m\Delta x)b_m^\hbar(\hat x)\| \cdot \|\hat z^n\hat z^m -\hat z^{n+m} \|_N\leq |a_nb_m|_\infty \lim_{N\to\infty} \|\hat z^n\hat z^m -\hat z^{n+m} \|_N=0 \,,    \\[5mm]
      \displaystyle       \lim_{N\to\infty}\|\hat z^{n+m}\|\cdot\|a^\hbar_n(\hat x+m\Delta x)\|\cdot \|b_m^\hbar(\hat x)-b_m^\hbar(\hat x+n\Delta x)  \|       \leq |a_n|_\infty |n b_m'|_\infty\lim_{N\rightarrow \infty} |\Delta x| =0\,,
             \end{array}
             $$
             where
             $$
             |nb'_m|_\infty=\lim_{N\to \infty}\left\|\frac{b_m^\hbar(\hat x)-b_m^\hbar(\hat x+n\Delta x) }{\Delta x}\right\|\,.
             $$
             Hence,
             $$
              \|(\widehat{z^na_n} ) (\widehat{z^mb_m})- \widehat{(z^{n+m}a_nb_m)}\|_N  \to 0\quad \mbox{as}\quad N\to\infty\,.  
              $$
              
 \paragraph{Dirac's condition.} Using the subadditivity of the norm, it suffices to consider quantum observables of the form (\ref{za}). 
 For a fixed integer $l\geq 0$, define the subspace of ``bulk states'' $\hh_{N}^l=\mathrm{span}\big\{e^{\frac{2\pi i n p}{M}}\big\}_{n=l+1}^{N-l}\subset \hh_N$.  Then, for all $n,m\leq l$, one has 
 \begin{equation}
    \hat{z}{}^n\hat{z}{}^m\psi= \hat{z}{}^m\hat{z}{}^n\psi= \hat{z}{}^{n+m}\psi\,,   \qquad \forall \psi\in \hh_N^l\,.
    \end{equation}
 Moreover,
 \begin{equation}
 \lim_{\hbar \rightarrow 0}\frac{\dim\hh_N^l}{\dim\hh_N}=\lim_{\hbar \rightarrow 0}\frac{LM-4\pi \hbar l}{LM}=1\,.
 \end{equation}
 The commutator of the quantum observables (\ref{za}) takes the form
\begin{equation}
     [\widehat{z^na_n}, \widehat{z^mb_m}]=\hat z^n\hat z^ma_n^\hbar(\hat x+m\Delta x)b_m^\hbar(\hat x)-\hat z^m\hat z^nb_m^\hbar(\hat x+n\Delta x)a_n^\hbar(\hat x)\,. 
\end{equation}
 Since $\hh^l_N$ is invariant under $\hat x$, this can be written as
 \begin{equation}\label{com}
     [\widehat{z^na_n}, \widehat{z^mb_m}]\psi= \hat{z}{}^{n+m}\big(a_n^\hbar(\hat x+m\Delta x)b_m^\hbar(\hat x)-b_m^\hbar(\hat x+n\Delta x)a_n^\hbar(\hat x)  \big)\psi
     \end{equation}
   for all   $\psi\in \hh_N^l $. 
   
   On  the classical side, one computes
$$
\{z^na_n,z^mb_m\}=z^{n+m}\frac{2\pi i}{M}(ma'_n b_m-nb'_ma_n)\,,
$$
and quantizing this yields
$$
\widehat{\{z^na_n,z^mb_m\}}=\hat z^{n+m}\frac{2\pi i}{M}\big(ma'^\hbar_n(\hat x+m\Delta x) b^\hbar_m(\hat x +n\Delta x)-nb'^\hbar_m(\hat x+n\Delta x)a_n(\hat x+m\Delta x)\big) \,.
$$
Applying this operator to a normalized state $\psi\in \hh_N^l$ and subtracting it from (\ref{com}) multiplied by $(i\hbar)^{-1}$, we obtain
$$
\begin{array}{c}
\displaystyle \big\|\big((i\hbar)^{-1}[\widehat{z^na_n}, \widehat{z^mb_m}] -\widehat{\{z^na_n,z^mb_m\}}\big)\psi\big\|\\[7mm]
\displaystyle \leq\big\|\hat z^{n+m}\big\| \cdot\frac{2\pi}{M}\left\|\frac{a_n^\hbar(\hat x+m\Delta x)-a_n^\hbar(\hat x)}{\Delta x}b_m^\hbar(\hat x)- \frac{b_m^\hbar(\hat x+n\Delta x)-b_m^\hbar(\hat x)}{\Delta x}a_n^\hbar(\hat x) \right.\\[7mm]
\displaystyle -ma'^\hbar_n(\hat x+m\Delta x) b^\hbar_m(\hat x +n\Delta x)+nb'^\hbar_m(\hat x+n\Delta x)a^\hbar_n(\hat x+m\Delta x)\Big\|\\[7mm]
\displaystyle \leq\frac{2\pi}{M}\left\|\left(\frac{a_n^\hbar(\hat x+m\Delta x)-a_n^\hbar(\hat x)}{\Delta x}-ma_n'^\hbar(\hat x)\right)b_m^\hbar(\hat x)- \left(\frac{b_m^\hbar(\hat x+n\Delta x)-b_m^\hbar(\hat x)}{\Delta x}-nb'^\hbar_m(\hat x)\right)a_n^\hbar(\hat x) \right.\\[7mm]
\displaystyle -ma'^\hbar_n(\hat x+m\Delta x) \big(b^\hbar_m(\hat x +n\Delta x)-b^\hbar_m(\hat x)\big)+nb'^\hbar_m(\hat x+n\Delta x)\big( a^\hbar_n(\hat x+m\Delta x) - a^\hbar_n(\hat x)\big)\Big\|\\[7mm]
\displaystyle\leq\frac{2\pi}{M}\left|\frac{a_n^\hbar( x+m\Delta x)-a_n^\hbar( x)}{\Delta x}-ma_n'^\hbar( x)\right|_\infty |b_m^\hbar|_\infty+ 
\frac{2\pi}{M}\left|\frac{b_m^\hbar(x+n\Delta x)-b_m^\hbar( x)}{\Delta x}-nb'^\hbar_m( x)\right|_\infty |a_n^\hbar|_\infty \\[7mm]
\displaystyle+\big| ma'^\hbar_n( x+m\Delta x)\big |_\infty \big|b^\hbar_m(x +n\Delta x)-b^\hbar_m( x)\big|_\infty+\big|nb'^\hbar_m( x+n\Delta x)\big|_\infty\big| a^\hbar_n(x+m\Delta x) - a^\hbar_n( x)\big|_\infty\,.
\end{array}
$$
Here we used (\ref{zf}). In the last expression, each term vanishes individually as $\hbar\to 0$ and $N\to \infty$. Hence, for any pair of classical observables (\ref{ab}) with $n,m\leq l$, choosing $\hh^{a,b}_N=\hh^l_N$ ensures the Dirac condition.

Turning to Proposition \ref{P2}, we consider an interval $[u,v]$ with the boundary points $u$ and $v$ distinct from zero. 
The finite-dimensional state space $\hh_N\subset \hh$ is spanned by the Fourier harmonics 
\begin{equation}
\begin{array}{c}
\psi_n(p)=e^{\frac{2\pi inp}{M}}\,,\qquad N_1\leq n\leq N_2\,,\\[5mm]
\displaystyle N_1=\left[\frac{uM}{2\pi \hbar }\right]\,,\qquad N_2=\left[\frac{vM}{2\pi \hbar }\right] \,, \qquad N=N_2-N_1\,.
\end{array}
\end{equation}
A general state of $\hh$ is represented by a  Fourier series 
\begin{equation}
    \psi=\sum_{k\in \mathbb{Z}}c_k \psi_k\,,
\end{equation}
with $c_k\to 0 $ as $|k|\to \infty$. Define the operator measuring the deviation from classical multiplication:
$$
A:=(\widehat{z^na_n} ) (\widehat{z^mb_m})- (\widehat{z^{n+m}a_nb_m})\,.
$$
To control the “edge effects,” we orthogonally decompose $\hh_N$ into the edge and bulk states:
\begin{equation}\label{od}
\begin{array}{c}
    P_N\psi=\psi'+\psi''\,,\\[5mm]
    \displaystyle \psi'=\sum_{k=N_1}^{N_1+m} c_k\psi_k+\sum_{k=N_2-m}^{N_2} c_k \psi_k  \,,\qquad 
    \psi''=  \sum_{k=N_1+m+1}^{N_2-m-1}   c_k\psi_k \,.  
    \end{array}
     \end{equation}
Applying the triangle inequality, we find
\begin{equation}
 \|A P_N\psi\big\|  \leq\sum_{k=N_1}^{N_1+m}   |c_k|\cdot \|A\psi_k\big\|
    + \sum_{k=N_2-m}^{N_2}   |c_k|\cdot  \|A\psi_k\big\|  
     +     \|A\psi''\big\|     \,,
     \end{equation}
which can be further bounded as
     \begin{equation}
\|A P_N\psi\big\|\leq 2|a_n|_\infty |b_m|_\infty\left(\sum_{k=N_1}^{N_1+m}   |c_k|+\sum_{k=N_2-m}^{N_2} |c_k|\right)    +
     \|A \psi''\big\|  \,.
     \end{equation}
The edge contributions  vanish  in the classical limit, since $|c_k|\rightarrow 0$ as $|k|\rightarrow \infty$. The bulk term $\|A\psi''\|$
is estimated as in the above derivation of Dirac’s condition. Using Eq. (\ref{dif}), we obtain:
\begin{equation}
\begin{array}{c}
    \|A\psi''\|=\big\|\hat z^{n+m}a^\hbar_n(\hat x+m\Delta x) \big[b_m^\hbar(\hat x)-b_m^\hbar(\hat x+n\Delta x) \big ]\psi''\big\| \\[5mm]
   \displaystyle \leq \|\psi''\|\cdot \big\|a^\hbar_n(\hat x+m\Delta x)\big\|\cdot    \left\|\frac {b_m^\hbar(\hat x)-b_m^\hbar(\hat x+n\Delta x)}{\Delta x} \right\|  \cdot |\Delta x|  \,,\end{array}
    \end{equation}
so that 
\begin{equation}
    \lim_{N\to \infty} \|A\psi''\|\leq \|\psi''\| \cdot |a_n|_\infty |n b'_m|_\infty \lim_{N\to \infty} |\Delta x|=0\,.
\end{equation}
Thus, combining the edge and bulk contributions, we obtain
$
    \lim_{N\to\infty}\|AP_N\psi\|=0
$,
which establishes the first limit in Proposition \ref{P2}.
The second limit follows by the same argument using  decomposition \ref{od}.

 \end{document}